\documentclass[letterpaper]{article} 
\usepackage{aaai25}  
\usepackage{times}  
\usepackage{helvet}  
\usepackage{courier}  
\usepackage[hyphens]{url}  
\usepackage{graphicx} 
\usepackage{natbib}  
\usepackage{caption} 

\usepackage{algorithm}
\usepackage{algorithmic}

\usepackage{newfloat}
\usepackage{listings}
\usepackage{graphicx}
\usepackage{amsfonts,amsmath,amsthm}
\usepackage{cleveref}

\usepackage{todonotes}
\presetkeys{todonotes}{inline}{}

\DeclareMathOperator*{\argmax}{arg\,max}
\usepackage{enumerate}

\usepackage{multirow}
\usepackage{algorithm}
\usepackage{algorithmic} %
\DeclareCaptionStyle{ruled}{labelfont=normalfont,labelsep=colon,strut=off} 
\floatstyle{ruled}
\newfloat{listing}{tb}{lst}{}
\floatname{listing}{Listing}
\newtheorem{theorem}{Theorem}[section]
\newtheorem{corollary}[theorem]{Corollary}
\newtheorem{proposition}[theorem]{Proposition}
\newtheorem{lemma}[theorem]{Lemma}

\newtheorem{conjecture}[theorem]{Conjecture}

\theoremstyle{definition}
\newtheorem{definition}[theorem]{Definition}

\usepackage{bibentry}

\begin{document}
\title{Understanding EFX Allocations: Counting and Variants}
\author {
    Tzeh Yuan Neoh\textsuperscript{\rm 1,2},
    Nicholas Teh\textsuperscript{\rm 3}
}
\affiliations {
    \textsuperscript{\rm 1}Institute of High Performance Computing, Agency for Science, Technology and Research, Singapore\\
        \textsuperscript{\rm 2}Centre for Frontier AI Research, Agency for Science, Technology and Research, Singapore\\
    \textsuperscript{\rm 3}University of Oxford, UK\\
    neohty@cfar.a-star.edu.sg, nicholas.teh@cs.ox.ac.uk
}

\maketitle

\begin{abstract}
Envy-freeness up to any good (EFX) is a popular and important fairness property in the fair allocation of indivisible goods, of which its existence in general is still an open question. In this work, we investigate the problem of determining the minimum number of EFX allocations for a given instance, arguing that this approach may yield valuable insights into the existence and computation of EFX allocations. We focus on restricted instances where the number of goods slightly exceeds the number of agents, and extend our analysis to weighted EFX (WEFX) and a novel variant of EFX for general monotone valuations, termed EFX+. In doing so, we identify the transition threshold for the existence of allocations satisfying these fairness notions. Notably, we resolve open problems regarding WEFX by proving polynomial-time computability under binary additive valuations, and establishing the first constant-factor approximation for two agents.
\end{abstract}

\section{Introduction}

The field of \emph{fair division} models many real-world problems such as course allocation \cite{budish2012courseallocation}, inheritance division \cite{pratt1990inheritance}, and divorce settlements \cite{BramsTa96}  (also refer to survey by \citet{amanatidis2023fairdivisionprogress}).
Among the numerous fairness properties that have been proposed, \emph{envy-freeness (EF)} (and its approximations or variants) remains one of the most widely studied in the field. 
Intuitively, an allocation is said to be envy-free if every agent values his bundle of goods at least as much as he values any other agent’s bundle of goods.
However, in the setting with indivisible goods, a complete envy-free allocation may fail to exist.
Hence, many works consider relaxations of envy-freeness. 

Perhaps the most prominent of these relaxations are \emph{envy-freeness up to any good (EFX)} \cite{caragiannis2019unreasonable} and \emph{envy-freeness up to one good (EF1)} \cite{lipton2004approximately}. 
An allocation satisfies EFX if any envy that an agent has towards another agent can be eliminated by removing \emph{any} good in the latter agent's bundle, and it satisfies EF1 if any envy that an agent has toward another agent can be eliminated by removing \emph{some} good in the latter agent's bundle (refer to Section \ref{sec:preliminaries} for the formal definitions).
The argument in favor of EFX as the preferred notion of fairness lies in the fact that EFX is considerably stronger compared to EF1. 
However, while EF1 allocations always exists and can be found in polynomial time (even when agents have arbitrary monotone valuations over the goods) \cite{Budish11,lipton2004approximately}, the existence of EFX remains one of the biggest open problems in the field of fair division \cite{caragiannis2019unreasonable,Moulin19,procaccia2020enigmatic}.  

Despite substantial efforts, the existence of EFX allocations have only been proven for several restricted instances. These include, but are not limited to 
(1) when $n=2$ \cite{plaut2020almost}, 
(2) when $n=3$ with at least one additive agent \cite{akrami2023efxrainbow}, 
(3) when $m = n+3$ \cite{mahara2023extension}, and
(4) under identical valuations \cite{plaut2020almost}.
Moreover, many previous algorithms have been shown to not extend to more general cases \cite{chaudhury2021little,mahara2023extension}. 

In this work, we ask a more general question: \emph{How many EFX allocations are there?} 
If there is an instance where EFX fail to exist, the answer is clearly $0$. 
However, for instances where EFX allocation(s) always exists, the answer to this question is substantially unclear. 
This question was first explored in \citet{SUKSOMPONG2020606} where they show that while there is always an exponential (with respect to the number of goods) number of EF1 allocations, there could be far fewer EFX allocations. 
In particular, when there are two agents, there can be as few as two EFX allocations, no matter the number of goods.
In this paper, we study the problem of counting the minimum number of EFX allocations for several restricted instances, and for two variants of EFX in more general instances.

We posit that pursuing this line of inquiry may uncover valuable insights regarding (the existence/computation of) EFX allocations, and make some progress towards resolving ``fair division's most enigmatic question'' \cite{procaccia2020enigmatic}:
\begin{itemize}
    \item Establishing tighter upper bounds on the minimum number of EFX allocations allows us to identify `hard' instances for EFX. 
    These `hard' instances can be used either as a first step to find an instance where an EFX allocation fails to exist, or to design algorithms tackling these `hard' instances. 
    \item Establishing tighter lower bounds on special cases allows us to better understand the kind of EFX allocations that exist in those special cases.
\end{itemize}
Moreover, in practice, it may be desirable to have differing EFX allocations. 
The psychological phenomena of \emph{single-option aversion} suggests that humans may be uncomfortable with decision-making if they are only presented with one option, even if it is an option they really like \cite{mochon2013single}. 
Beyond psychology, consider an instance with two agents and three goods, where each agent values every good identically. 
Clearly, a complete EF allocation is not possible.\footnote{A \emph{complete} allocation is one where all goods are allocated. This assumption is common in the literature; otherwise, one could leave all goods unallocated and it is trivially `fair'.}
The best that any \emph{deterministic} mechanism can achieve in this instance is an EFX allocation that allocates any two goods to an agent and one good to the other. 

However, consider a \emph{randomized} mechanism that gives the first agent two goods with a probability of half, and one good otherwise. 
By having more than one EFX allocation, we can design a mechanism that only outputs EFX allocations and is EF in expectation! 
This idea is also captured in the \emph{Best-of-Both-Worlds (BoBW)} framework where the goal is to design a mechanism that has stronger ex-ante guarantees despite relatively weaker ex-post guarantees. 
Under this framework, it was shown that there exists an algorithm that is ex-ante EF and ex-post EF1 \cite{aziz2020simultaneously,freeman2020best}. 
Such algorithms implicitly hinges on the fact there are numerous EF1 allocations that one can randomize over, so as to achieve a fair outcome in expectation.

\subsection{Our Contributions}
We consider the standard fair division model with indivisible goods.
Our results are summarized in Table \ref{tab:summary}.
\begin{table*}[h]
\centering
\begin{tabular}{c|c||c|c|c|c}
Property & Instance  & $m \leq n$ & $m = n+1$ & $ m = n+2$ & $m \geq n+3$\\ \hline \hline
 \multirow{2}{*}{EFX} & General & $\frac{n!}{(m-n)!}$ &  $n$ & UB: $n^2$, LB: $n$& UB: $n!$ \\ \cline{2-6}
 & Identical Valuations & $\frac{n!}{(m-n)!}$ & \multicolumn{3}{|c}{$n!$}\\ \hline
 \multirow{3}{*}{WEFX}& $n = 2$, Binary Additive &  2 & \multicolumn{3}{c}{1} \\ \cline{2-6}
  &$n = 2$, Binary Submodular & 2& 1 & \multicolumn{2}{c}{0} \\ \cline{2-6}
    & $n = 2$, Restricted Additive & 2& 1 & \multicolumn{2}{c}{0} \\ \hline
 \multirow{2}{*}{EFX+}& $n = 2$& \multicolumn{4}{c}{2} \\ \cline{2-6}
  & $n \geq 3$&  $\frac{n!}{(m-n)!}$ & \multicolumn{3}{c}{0} \\
\end{tabular}
\caption{Summary of our results: counting the minimum number of allocations satisfying the fairness property in various special cases. The upper bounds (UB) and lower bounds (LB) are provided in cases where the bound is not tight. Note that the lower bounds apply for general monotone valuations, but the examples constructed for the upper bounds only assume additive valuations, thereby strengthening our results.}
\label{tab:summary}
\end{table*}

In Section \ref{sec:EFX}, we first show that given an instance, counting the number of EFX allocations is $\#P$-complete. We then show that the minimum number of EFX allocations is a function of both the number of agents ($n$) and number of goods ($m$) by investigating the restricted instance with few goods. 
We elucidate that this setting is non-trivial by showing that if an EFX allocation always exists when $m = n + \omega(1)$, then an EFX allocation always exists in general. 
For counting the minimum number of solutions, we show that there is a minimum of $\frac{n!}{(n-m)!}$ EFX allocations when $m\leq n$ and a minimum of $n$ EFX allocations when $m = n+1$. 
When $m = n+2$, we provide a lower bound of $n$ and an upper bound of $n^2$ for the minimum number of EFX allocations. 
Lastly, we show that there can always be as few as $n!$ EFX allocation for any number of goods. 

In Section \ref{sec:WEF}, we investigate weighted envy-freeness up to any good (WEFX), a generalization of EFX to the setting where agents have possibly differing entitlements. 
We show that under binary additive valuations, there always exists a WEFX allocation; but when $m \geq 3$, there could be as few as a single WEFX allocation. 
Moreover, we resolve an open question by providing a polynomial-time algorithm for computing WEFX and PO allocations for any number of agents under binary valuations.
However, for the more general binary submodular or restricted additive valuation classes, we show that even after relaxing WEFX to a weaker variant---weak WEFX (WWEFX), a WWEFX allocation may still fail to exist in these settings. Lastly, under additive valuations, we show that there always exists a $\frac{1}{4}$-WEFX allocation when $n = 2$, improving upon the approximation factor previously shown by \citet{hajiaghayi2023wef}, and resolving an open question posed by \citet{suksompong2025weightedreview}.

In Section \ref{sec:EFX+}, we investigate EFX+, a novel variant that is identical to EFX under additive valuations, but incomparable to EFX under general monotone valuations. 
    
\subsection{Further Related Work}
Numerous variants (relaxations) of EFX (which is stronger than EF1) have been proposed and studied, such as a multiplicative relaxation of EFX ($\alpha$-EFX) \cite{amanatidis2020multiple}, envy-freeness up to one less-preferred good (EFL) \cite{barman2018groupwise}, envy-freeness up to a random good (EFR) \cite{farhadi2021almost}, and epistemic EFX \cite{caragiannis2022existence}.
Several other works consider the notion of EFX with charity, a relaxation of EFX where not all goods need to be allocated \cite{caragiannis2019efxcharity,chaudhury2021little}.

EFX has also been studied in the context of indivisible chores division. As the chores setting is known to be more complex in general, the existence of EFX is also an open problem there. There has been work showing that EFX for chores exists in restricted instances \cite{kobayashi2023efx,miao2023efx}, in constrained settings \cite{Elkind2024choresbudget}, as well as upper bounds on approximate EFX \cite{bhaskar2020approximate,zhou2024approximately}. 

\section{Preliminaries} \label{sec:preliminaries}
For each positive integer $z$, let $[z] := \{1,\dots,z\}$.
An instance consist of a set $N = [n]$ of $n$ \emph{agents} and a set $G = \{g_1, \dots g_m\}$ of $m$ \emph{goods}. 
Each agent $i \in N$ has a \emph{valuation function} $v_i : 2^G \rightarrow \mathbb{R}_{\geq 0}$. 
For each $i \in N$, we assume 
that the valuation function $v_i$ is  \emph{monotone} (i.e., $S\subseteq T \implies v_i(S) \leq v_i(T)$) and $v_i(\emptyset) = 0$;
$v_i$ is said to be \emph{additive} if for any subset of goods $S \subseteq G$, $v_i(S) = \sum_{g \in S} v_i(g)$.\footnote{In general, we do not assume additivity for agents' valuation functions, but will use this in our counterexamples and upper bounds, which strengthens our negative results.}
For notational simplicity, we sometimes write $v_i(g)$ instead of $v_i(\{g\})$.
Let the \emph{valuation profile} of agents be $\mathbf{v} = (v_1,\dots,v_n)$.
An \emph{instance} $I = (N',G',\mathbf{v}')$ is defined by a set of agents $N' \subseteq N$, a set of goods $G' \subseteq G$, and valuation profile $\mathbf{v}'$ corresponding to the set of agents and goods.
An \emph{allocation} $\mathcal{A} = (A_1, \dots, A_n)$ is a partition of $G$ such that each agent $i \in N$ receives the bundle $A_i$.
Throughout this work, we assume the eventual goal is to obtain a \emph{complete} allocation, i.e., all goods are eventually allocated.

Next, we formally define EF and EFX, as follows.
\begin{definition}[EF]
    An allocation $\mathcal{A} = (A_1,\dots,A_n)$ is \emph{envy-free (EF)} if for all agents $i,j \in N$, $v_i(A_i) \geq v_i(A_j)$.
\end{definition}

 \begin{definition}[EFX] \label{def:EFX}
     An allocation $\mathcal{A} = (A_1,\dots,A_n)$ is \emph{envy-free up to any good (EFX)} if for all agents $i,j \in N$ and all $g \in A_j$, $v_i(A_i) \geq v_i(A_j \setminus \{g\})$.
 \end{definition}
It is instructive to note that the above definition of EFX is sometimes referred to as EFX$_0$ (introduced by \citet{KYROPOULOU2020110}), with EFX being originally defined as a weaker version that only drops nonzero-valued goods \cite{caragiannis2019unreasonable}. However, we adopt the stronger definition and refer to it simply as EFX, in line with several recent works (see the survey of \citet{amanatidis2023fairdivisionprogress}).
Notably, this also provides us with stronger results: in terms of existence, an EFX allocation exists if and only if an EFX$_0$ allocation exists \cite{chaudhury2021improving}.

We assume that the reader is familiar with basic notions of classic complexity theory~\cite{papadimitriou_computational_2007}.
All omitted proofs can be found in the appendix.

We also make use of the concept of an \emph{envy graph} \cite{lipton2004approximately}.
The envy graph $\mathcal{G} = (N,E)$ for a (partial) allocation $\mathcal{A}$, is a directed graph where each agent is represented by a vertex. 
Moreover, for agents $i,j \in N$, $(i,j) \in E$ if agents agent $i$ envies agent $j$. 
If there exists a cycle in the envy graph, we can perform a cyclic exchange of bundles, also known as \emph{envy cycle elimination}.
This process results in a new allocation where every agent swaps their bundle in exchange for an envied bundle, and consequently, no agent is worse off. 
Crucially, in our setting, if the initial allocation is EFX, after envy cycle elimination, the new allocation will remain EFX.

We first consider the standard EFX property.

\section{EFX} \label{sec:EFX}
Going back to our motivation on BoBW outcomes, we first look into the case with two additive agents. When $n=2$, the minimum number of EFX allocations is $2$ \cite{SUKSOMPONG2020606}. By randomizing over these EFX allocations, an ex-ante EF and ex-post EFX outcome (lottery over allocations) always exists, which we show with the following result.
\begin{proposition}\label{prop:bobw}
    When $n = 2$,
    an ex-ante \emph{EF} and ex-post \emph{EFX} outcome always exists.
\end{proposition}

However, despite the usefulness of knowing the number of EFX allocations as motivated earlier, in general, given an instance, counting the number of EFX allocations is $\#P$-complete.
This might be the case even in instances where finding an EFX allocation is easy.

We begin with a useful observation. 
    
    \begin{lemma}
        Let $a_0, \dots, a_n$ be a sequence of non-negative integers such that $\sum_{i = 0}^n a_i \leq n!$. 
        If we know the value $P = \sum_{i = 0}^n a_i \cdot i^k$ for $k \geq \frac{n\ln(n)}{\ln (n/(n-1))}$, then we can deduce $a_n$.

    \end{lemma}
    \begin{proof}
        We will show that $n^k > (n-1)^k \cdot n!$ and thus $\left\lfloor \frac{P}{n^k} \right\rfloor = a_n$. 
        Note that $k > \frac{n\ln(n)}{\ln (n/(n-1))} \iff    
k \ln\frac{n}{n-1} > n\ln(n) \implies k \ln\frac{n}{n-1} > \ln (n!) \iff k \ln n > k \ln (n-1) \ln(n!)\iff  n^k > (n-1)^k \cdot n!   $ 
\end{proof}
Then our main result is as follows.
\begin{theorem}
    Given an instance, counting the number of \emph{EFX} allocations is $\#$P-complete in general.
\end{theorem}
\begin{proof}
We reduce from the $\#P$-complete problem of counting perfect matchings in a bipartite graph \cite{valiant1979complexity}. 
The problem is as follows: given  $G = (X \cup Y, E)$, with $X = \{x_1, \dots, x_n\}$ and $Y = \{y_1, \dots, y_n\}$, count the number of perfect matchings in $G.$

First, if there exists some node with no neighbors, we can return $0$ as there cannot be any perfect matching. 
Thus, we assume that all nodes in $X$ contain at least $1$ neighbor. 
Next, we construct an instance $\mathcal{I}$ with $n$ agents and $n + k$ goods. Let $v_{i}(g_j) = 1$ if $j \leq n$ and $(x_i,y_j)\in E$, $v_{i}(g_j) = 0.5$ if $j \leq n$ and $(x_i,y_j)\notin E$,  and $v_{i}(g_j) = 0$ otherwise. 
Then, our agents correspond to $X$, our first $n$ goods correspond to $Y$ and the remaining $k$ goods are valued at $0$.

Let $\mathcal{A}$ be the partial EFX allocation over the first $n$ goods, and consider the envy graph induced by $\mathcal{A}$.
Clearly, if any agent receives more than one good in $\mathcal{A}$, then $\mathcal{A}$ (and any complete allocation that extends $\mathcal{A}$) is not EFX.
If $\mathcal{A}$ corresponds to a perfect matching, then for all agents $i,j \in [n]$,  $v_i(A_i) = 1$ and $v_i(A_j) \leq 1$, and there is no envy. 
If $\mathcal{A}$ does not correspond to a perfect matching, we will show that there must be an envied agent. 
Suppose $\mathcal{A}$ does not correspond to a perfect matching. 
Then there must exist an agent $i \in N$ such that $v_i(A_i) = \frac{1}{2}$. 
Furthermore, as each vertex in $X$ has a neighbor, there is a $l \in [n]$ such that $(x_i,y_l) \in E$ and thus $v_i(y_l) = 1$. 
Suppose agent $j$ receives good $l$ under $\mathcal{A}$. Then $v_i(A_j) \geq 1$ and agent $j$ is envied. 

Thus, if $\mathcal{A}$ has no envied agents, then we can allocate the remaining $k$ goods freely and still obtain an EFX allocation. 
This gives us $n^k$ complete EFX allocations that has partial allocation $\mathcal{A}$ over the first $n$ goods. 
Conversely, suppose agent $i$ envies agent $j$ under $\mathcal{A}$. 
Then, if we allocate one or more of the remaining goods to $j$, the resulting allocation cannot be EFX. 
After all, as all the remaining $k$ goods are valued at $0$, agent $i$ still has a value of $\frac{1}{2}$ for their bundle and agent $i$ has a value of $1$ for agent $j$'s bundle even after removing a good from it. 
Thus, with $i$ envied agents, 
there are at most $(n-i)^k$ complete EFX allocations that has the partial allocation $\mathcal{A}$ over the first $n$ goods.

Let $a_i$ be the number of partial EFX allocations over the first $n$ goods with $i$ agents not envied in the envy graph. 
There are a total of $\sum_{i=0}^{n} a_i \cdot i^k$ complete EFX allocations. 
Since $a_i$ is always non-negative and $\sum_{i=0}^{n} a_i < n!$, by our observation above, we can set $k$ appropriately and deduce $a_n$.
Moreover, as $a_n$ is the number of partial EFX allocations over the first $n$ goods such that all agents are unenvied, $a_n$ is also the number of perfect matchings under $G$. 
Finally, when  $y > 0$, we have that $\ln (1 + y) \geq y - 2y^2$, and $k \in \mathcal{O}(n^3)$. Our reduction can thus be done in polynomial time.
\end{proof}

Thus, for the remainder of this section, we focus on restricted instances with few goods. We first note that the existence of EFX is only known for the case where $m \leq n+3$. This was shown in a rather technical proof involving \emph{champion graphs} and \emph{lexicographic potential functions} \cite{mahara2023extension}. We posit that even when there are few goods, determining the existence of an EFX allocation is non-trivial: because if an EFX allocation always exists when $m = n + \lfloor f(n) \rfloor$ for any  $f(n) \in \omega(1)$\footnote{A function $f(n)$ is in $\omega(1)$ if and only if $\lim_{n \to \infty} f(n) = \infty$.}, then an EFX allocation always exists in general.
\begin{theorem}
\label{small is hard} 
    If there exists a function $f(n) \in \omega(1)$, such that there is always an \emph{EFX} allocation for any $m = n + \lfloor f(n) \rfloor$, then there is an \emph{EFX} allocation for every $n,m$.
\end{theorem}

Next, we proceed to show our main results of this section, beginning with the simple cases when $m \leq n$ and when $m = n+1$ where we show a tight bound. Note that the lower bounds apply for general monotone valuations, but the examples constructed for the upper bounds only assume additive valuations. 

\begin{proposition}\label{thm: counting: n<=m}
    When $m\leq n$, every instance has at least $\frac{n!}{(n-m)!}$ \emph{EFX} allocations, and there exists an instance with at most $\frac{n!}{(n-m)!}$ \emph{EFX} allocations.
\end{proposition}
\begin{proof}
    For the lower bound, it is easy to see that any allocation where each agent receives at most one good is trivially EFX. 
    There are ${n \choose m} \cdot m! = \frac{n!}{(n-m)!}$ such allocations where every agent receives at most one good.
    
    For the upper bound, consider the case where every agent has a positive valuation for every good. 
    In this case, it is necessary that all agents receive at most one good for the allocation to be EFX. 
\end{proof}

\begin{proposition}\label{thm: counting: n=m+1}
    When $m = n+ 1$, every instance has at least $n$ \emph{EFX} allocations, and there exists an instance with at most $n$ \emph{EFX} allocations.
\end{proposition}

\begin{proof}    

    We first prove the lower bound. 
    Fix some arbitrary ordering of the agents, and label agents $1,\dots,n$ in that order.
    Then, let each agent (in increasing order) pick their favorite good out of the remaining unallocated goods, and let last agent $n$ get the remaining unallocated good. 
    This yields an EFX allocation where every agent $i \in [n-1]$ receives one good and the agent $n$ receives two goods. 
    For all $i \in [n-1]$, setting agent $i$ to be the last agent in some ordering will yield an EFX allocation such that agent $i$ receives two goods, and with any other agent $j \in [n]\setminus \{i\}$ receiving one good. 
    Hence, there are at least $n$ different EFX allocations.
    
    For the upper bound, consider an instance where agents have valuations as follows: for each $i \in [n]$ and $j \in [n+1]$, $v_i(g_j) = 2$ if $i = j$, $v_i(g_j) = 1$ if $j = n+1$, and $v_i(g_j) = 0$ otherwise. In all EFX allocations, for all agents $i \in [n]$, agent $i$ will receive $g_i$. 
    There are exactly $n$ EFX allocations, differing only by which agent receives $g_{n+1}$.
\end{proof}
We now proceed to show an upper bound of $n^2$ and a lower bound of $n$ for the case of $m=n+2$, with the following two theorems---the first showing the upper bound, and the next, which is more involved, showing the lower bound.
\begin{theorem}\label{UB: counting: m = n+2}
    When $m = n+2$, there exist an instance with at most $n^2$ \emph{EFX} allocations. 
\end{theorem}
\begin{proof}
    Consider an instance where agents have valuations over goods defined as follows: for each $i \in [n]$ and $j \in [n+2]$, $v_i(g_j) = 3$ if $i = j$, $v_i(g_j) = 1$ if $j \in \{n+1,n+2\}$, and $v_i(g_j) = 0$ otherwise. 
    Suppose some agent $i \in [n]$ receives at least two goods (including good $g_j$, for some $j \neq i$). 
    If $j\leq n$, then agent $j$ will envy agent $i$ even after removing the least-valued good (in $j$'s view) from agent $i$'s bundle (which will not be $g_j$), and hence the allocation is not EFX. 
    Thus, without loss of generality, suppose $j = n+1$, i.e., agent $i$ receives $g_{n+1}$.
    We show that in every EFX allocation, for all agents $x \in [n]$, agent $x$ must receive $g_{x}$.
    Suppose for a contradiction there exists an EFX allocation $\mathcal{A}$ and an agent $x \in [n]$ such that $g_x \notin A_x$. 
    Let $g_x \in A_y$ for some other agent $y \in [n] \setminus \{x\}$. 
    From our previous observation, $A_y = \{g_x\}$.
    This means $v_y(A_y) = 0$.
    However, we know that $g_{n+1} \in A_i$, and hence $v_y(A_i \setminus \{g\}) \geq 1$, contradicting the fact that $\mathcal{A}$ is an EFX allocation.

    Thus, there are exactly $n^2$ EFX allocations, differing only by which agent receives the goods $g_{n+1}$ and $g_{n+2}$.
\end{proof}
\begin{theorem}\label{thm: counting: m = n+2}
    When $m = n+2$, every instance has at least $n$ \emph{EFX} allocations.
\end{theorem}
    We draw inspiration from the algorithms of \citet{amanatidis2020multiple} that together imply the existence of EFX allocations when $m = n+2$, with the following algorithm (Algorithm \ref{alg:efx_m=n+2}).
    \begin{algorithm}
    \caption{Returns an EFX allocation when $m=n+2$.}
    \begin{algorithmic}[1]
    \label{alg:efx_m=n+2}
    \newcommand{\Step}[2]{\STATE #1 \hfill$\triangleright$ \textit{#2}}
    \STATE \textbf{Input:} Set of agents $N$, set of goods $G$, valuation profile $\mathbf{v} = (v_1,\dots,v_n)$.
    \STATE Order the agents from $1$ to $n$.
    \STATE Initialize the partial allocation, $A \leftarrow (\emptyset \dots \emptyset)$
    \STATE Initialize a set of unallocated goods, $P \leftarrow G$
    \FOR{$i=1, \dots, n-1$}
    \STATE Let $g \in \argmax_{g' \in P} v_i(g')$ be a good in $P$ that agent $i$ values the most
    \STATE Allocate the good $g$ to agent $i$, $A_i \leftarrow \{g\} $
    \STATE Remove $g$ from the set of unallocated goods, $P \leftarrow P \setminus \{g\}$
    \ENDFOR
    \STATE Create \emph{virtual goods} $s , l$ such that $\forall i \in N$, $v_i(s) = \min_{g \in P} v_i(g)$ and $v_i(l) = v_i(P) - v_i(s)$ 
    \Step{Allocate good $l$ to agent $n$, $A_n \leftarrow \{l\}$}{State 1}
    \Step {Perform envy cycle elimination and allocate $s$ to the non-envied agent}{State 2}
    \STATE Allocate the final owner of $l$ her two most valued goods in $P$, and the final owner of $s$ the remaining good in $P$.
    \end{algorithmic}
    \end{algorithm}

    We then show that the above algorithm returns an EFX allocation.
    It is clear that the partial allocation at State 1 is EFX as agents $1$ to $n-1$ are allocated only one good each, and they weakly prefer their allocated good to any good in $P$. 
    Hence, after envy cycle elimination, the partial allocation at State 2 is also EFX, and there exists at least one non-envied agent (let it be agent $i$). 
    If agent $i$ has two goods (and hence $l$), all agents will weakly prefer the good they are allocated over any two goods in $P$. 
    Hence, no agent will envy agent $i$ (after dropping their least-valued good from agent $i$'s bundle) if they were to receive all goods in $P$. 
    If $i$ has one good $g$, all agents will weakly prefer their bundle over either $g$ or $s$. 
    Hence, no agent will envy agent $i$---after dropping their least-valued good from $i$'s bundle---if she were to receive both $g$ and $s$, and the allocation is EFX.

    Next, we show---with the following lemma---an important structure of some EFX allocation.
    
    \begin{lemma}\label{m,n+2,lemma 1}
        When $m = n+2$, for each agent $i \in N$, there exists an \emph{EFX} allocation $\mathcal{A}$ whereby each agent receives at least one good, and either 
            \begin{enumerate}
                \item $i$ receives three goods; or
                \item $i$ and some agent $j \neq i$ each receives two goods, and one of $\{i,j\}$ receives their most valued good in $A_i \cup A_j$. 
        \end{enumerate}
    \end{lemma}
    We then introduce the concept of a \emph{join-graph}, a directed graph where the vertices are the set of agents and the edges represents different EFX allocations.

 \begin{definition}[Join-graph]
     A \emph{join-graph} $G = (N,E)$ is a directed graph where an edge $(i,j) \in E$ if one of the following two criteria is satisfied.
     \begin{itemize}
        \item Criterion 1: $i \neq j$ and there is an EFX allocation $\mathcal{A}$ where
    \begin{itemize}
      \item each agent receives at least one good;
      \item agent $i$ and $j$ each receives two goods; and
      \item agent $j$ receives her favourite good in $A_i\cup A_j$.
    \end{itemize}
  \item Criterion 2: $i = j$ and there is an EFX allocation $\mathcal{A}$ where
    \begin{itemize}
      \item each agent receive at least one good; and
      \item agent $i$ receives three goods.
    \end{itemize}

\end{itemize}
 \end{definition}

 Then, we prove another lemma.
\begin{lemma}\label{lem:joingraph|E|}
    For the join-graph $G = (N, E)$ and any instance with $n+2$ goods, there exists at least $|E|$ \emph{EFX} allocations.
\end{lemma} 
Then, our final lemma below proves our result.
    \begin{lemma}
        For each connected component $G_x = (N_x, E_x)$ in the join-graph $G = (V,E)$,  $|E_i| \geq |N_i|$. Hence, $|E| \geq |N|$ and there are at least $n$ \emph{EFX} allocations.
    \end{lemma}
    \begin{proof}
        Suppose there is a connected component $G_x = (N_x, E_x)$ such that $|E_x|<|N_x|$. 
        First observe that in a such connected component $G_x$, if $|E_x|<|N_x|$ then there are no in cycles in $G_x$.
        This means that there are no self-loops, and thus no EFX allocation in which any agent in $N_x$ receives three goods. 
        There also must be a source agent $i$ such that for all agents $j$, $(j,i) \notin E$. 
        From Lemma \ref{m,n+2,lemma 1}, in an instance with $m = n+2$, there is at least one incoming or outgoing edge for every agent in the join-graph.
        As $i$ has no incoming edge in the join-graph, then there must be an agent $j$ such that $(i,j) \in E$. 
    
        Let $\mathcal{A}$ be the EFX allocation that corresponds to $(i,j)$ and let $A_i = \{x_1,x_2\}$ and $A_j = \{y_1,y_2\}$.
        Without loss of generality, let $v_j(y_1) \geq v_j(y_2)$.
        Also let $P = \{x_1,x_2,y_2\}$ and create virtual goods $s,l$ such that $v_i(s) = \min_{g \in P} v_i(g)$ and $v_i(l) = v_i(P) - v_i(s)$. 
        Consider the partial allocation $\mathcal{A}'$ where $j$ receives only $y_1$, $i$ receives $l$, and all other agents receives the same good as they did in $\mathcal{A}$. 
        Note that $\mathcal{A}'$ is EFX. 
        As $(i,i) \notin E$, there is an agent $k$ that envies agent $i$ under the partial allocation $\mathcal{A}'$; otherwise, we can allocate $P$ to $i$ and still get an EFX allocation. 
    
        Now, there exists a path in the envy graph for $\mathcal{A}'$ from $k$ to a non-envied agent. 
        Let this path be $y_0, \dots y_\alpha$ where $y_0 = k$, $y_\alpha$ is a non-envied agent and $y_i$ envies $y_{i-1}$ for all $i \in [\alpha]$.
        Note that if $k$ is non-envied, $y_0 = y_\alpha = k$. 
        
        If there is a path from $k$ to a non-envied agent other than $i$, let $A'_{y_\alpha} = \{z\}$. 
        As $y_\alpha$ is non-envied, we can create an EFX allocation by allocating $s$ to $y_\alpha$. 
        Then, since $(y_\alpha, i) \notin E$, $i$ values $z$ the most in $P \cup \{z\}$. 
        Construct an allocation $\mathcal{B}$ by setting $B_{y_i} = A'_{y_{i-1}}$ for $i \in [\alpha]$, $B_{i} =\{s, z\}$, $B_{k} = l$ and allocating all other agents their good in $A'$. 
        Observe that $\mathcal{B}$ is EFX and $i$ receives their favourite good in $P \cup \{z\}$, $(k,i) \in E$, giving us a contradiction. 
    
        In the other case, the only path from $k$ to a non-envied agent is a path to $i$. 
        This is only possible if $y_{\alpha - 1}$ is only envied by agent $i$. 
        Hence, we can construct an allocation $\mathcal{B}$ by setting $B_{y_i} = A'_{y_{i-1}}$ for $i \in [\alpha -1]$, $B_{i} = A'_{y_{\alpha-1}} \cup{\{s\}}$, $B_{k} = l$ and allocating all other agents their good in $\mathcal{A}'$. 
        Observe that $\mathcal{B}$ is EFX and $i$ receives their favourite good in $B_i \cup B_k$, $(k,i) \in E$, giving us a contradiction.
    \end{proof}

To end this section, we show a tight bound for the special case of identical valuations, a well-studied setting in fair division \cite{barman2020identical,plaut2020almost,Choo2024}.

\begin{theorem} \label{thm: counting: identical}
    When $m\geq n$, under identical valuations, every instance has at least $n!$ EFX allocations, and there exists an instance with at most $n!$ \emph{EFX} allocations.
\end{theorem}
\begin{proof}
For the lower bound, we know that there always exists an EFX allocation under identical valuations \cite{plaut2020almost}. 
When $m \geq n$, there is always an EFX allocation $\mathcal{A}$ such that each agent receives at least one good. 
All allocations that are a permutation of $\mathcal{A}$ (i.e., bundles remain the same, but may be allocated to a different agent) also satisfies EFX, giving us at least $n!$ EFX allocations. 

For the upper bound, consider the instance where $v(g_j) = 1$ if $j\leq n- 1$ and $v(g_j) = 0$, otherwise. Then, there are only $n!$ EFX allocations. The proof of correctness for this upper bound is deferred to the appendix.
\end{proof}
Importantly, we obtain the following corollary, contrasting the number of EFX allocations with the minimum number of EF1 allocations which is known to be exponential in $m$ even for a fixed $n$ \cite{SUKSOMPONG2020606}.
\begin{corollary}
    For any $m$, the number of \emph{EFX} allocations can be as few as $n!$.
\end{corollary}

Next, we consider two variants of EFX, and study similar questions.
The first variant we will consider is an ``up to any good'' relaxation of \emph{weighted envy-freeness (WEF)}.

\section{Weighted EFX} \label{sec:WEF}
In recent years, there has been a growing number of works in weighted fair division, an extension of the standard fair division model to one where agents have differing entitlements \cite{AzizMoSa20,BabaioffEzFe21,BabaioffNiTa21,chakraborty2021weighted,ChakrabortySeSu22,HoeferScVa23,ScarlettTeZi23}.
In this model, each agent $i\in N$ has an additional \emph{weight} parameter $w_i > 0$ representing her entitlement.

This model would allow us to better capture settings such as inheritance division or divorce settlements where entitlements are typically unequal.
Moreover, weighted fair division also extends the well-studied setting of \emph{apportionment}, which is used in political systems to allocate parliamentary seats \cite{BalinskiYo01,Pukelsheim14}.
\citet{chakraborty2021weighted} primarily studied the generalization of EF1 to \emph{weighted EF1 (WEF1)}.
Weighted EFX (WEFX), a natural generalization of EFX, has also been looked into, but there has been strong negative results.

In this section, we first show the existence of a WEFX and \emph{Pareto-optimal (PO)} allocation under binary additive valuation for any number of agents, resolving an open question in the area. 
We also show that even in this restricted instance where WEFX allocations exists, there can be instances with only one WEFX allocation. 
In addition, we provide stronger negative results for the non-existence of WEFX, by further showing that a weaker version of WEFX (i.e., weak WEFX) also does not exist even in very restricted settings, thereby considerably tightening the existential gap known in the literature. 
We first state the definition of WEFX.\footnote{Note that we are actually considering WEFX$_0$, a stronger variant of WEFX that can even drop zero-valued goods.}

\begin{definition}[WEFX]
    An allocation $\mathcal{A} = (A_1,\dots,A_n)$ is \emph{weighted envy-free up to any good (WEFX)} if for all agents $i,j \in N$ and all $g \in A_j$,
    $\frac{v_i(A_i)}{w_i} \geq \frac{v_i(A_j \setminus \{g\})}{w_j}$.
\end{definition}
\citet{hajiaghayi2023wef} showed an impossibility result---that an WEFX allocation may not exist even for $n=2$ under additive valuations. 
An immediate question is then: are there any (more restricted) settings where a WEFX allocation exists? And if so, how many are there?

Indeed, we first show that in the restricted setting of binary additive valuations---which has also been well studied in the fair division literature \cite{aleksandrov2015onlinefoodbank,amanatidis2021mnwefx,bouveret2016conflict,freeman2019equitable,halpern2020binary,SuksompongTe22}---with any number of agents, not only does an WEFX allocation exist, but we can find an allocation satisfying both WEFX and PO.
This combination is surprisingly elusive---a weighted generalization of \emph{maximum Nash welfare} has been shown to not always satisfy WEF1 and PO under binary additive valuations.

We first state the definition of PO.
\begin{definition}[PO]
An allocation $\mathcal{A} = (A_1,\dots,A_n)$ is \emph{Pareto-optimal (PO)} if there does not exist another allocation $\mathcal{A}'$ such that $v_i(A'_i)\geq v_i(A_i)$ for all $i \in N$, and $v_i(A'_i) > v_i(A_i)$ for some $i \in N$.
\end{definition}

We show the existence of WEFX and PO allocations through two methods. 
Firstly, we show that by appropriately adapting the definiton of \emph{leximin} to account for the weights, a weighted leximin allocation (that maximizes a version of weighted egalitarian welfare function) is WEFX and PO. 
The benefit of this is that it is a \emph{welfare function}, which is widely studied in the literature, and may potentially offer additional fairness or efficiency guarantees. 
We then provide a matching-based algorithm that can compute a WEFX and PO allocation in polynomial time.

\begin{theorem}\label{thm:WEFX,n =2, PO + B,A}
    Under binary additive valuations, for any number of agents, a \emph{WEFX} and \emph{PO} allocation exists and can be computed in polynomial-time.
\end{theorem}

However, while an allocation satisfying WEFX exists in this restricted setting, there are not many of them, as illustrated with the following result. 

\begin{theorem}\label{1 allocation: WEFX}
    Under binary additive valuations with $n=2$, 
    \begin{itemize}
        \item when $m \leq 2$, every instance has at least $2$ \emph{WEFX} allocations, and there exists an instance with at most $2$ \emph{WEFX} allocations;
        \item when $m \geq 3$, every instance has at least $1$ \emph{WEFX} allocations, and there exists an instance with at most $1$ \emph{WEFX} allocation.
    \end{itemize}
\end{theorem}
A follow-up question is: does the existence of WEFX allocation hold for valuation functions that are more general than binary additive functions?
To address this, we consider two more general valuation classes commonly studied in the literature---namely \emph{binary submodular} and \emph{restricted additive} valuations.

A valuation profile $(v_1,\dots,v_n)$ is said to be \emph{binary submodular} (or \emph{matroid-rank}) if for all $i \in N, g \in G$ and $ S, T \subseteq G$, $v_i(g | S) = 0$ or $v_i(g | S) = 1$ and $v_i(S) + v_i(T) \geq v_i(S \cup T) + v_i(S \cap T)$   \cite{BabaioffEzFe21-dichotomous,BarmanVe21,BenabbouChIg21,GokoIgKa22,MontanariScSu24,SuksompongTe23}; 
and \emph{restricted additive} (or \emph{generalized binary}) if $v_1,\dots,v_n$ are additive and there exists a function $h:G\rightarrow\mathbb{R}_{\ge 0}$ such that $v_i(g)\in \{0,h(g)\}$ for all $i\in N$ and $g\in G$ \cite{AkramiReSe22,CamachoFePe23,elkind2025temporalfd}.
We show that considering the above-defined  generalizations of valuation functions beyond binary additive functions would yield negative results when $n=2$, and not only for WEFX, but for a weaker version of WEFX, called \emph{weak WEFX (WWEFX)}.
This weaker relaxation has also been studied for WEF1 \cite{chakraborty2021weighted,ChakrabortySeSu22}.

Our results essentially close the (WEFX) existential gap of existing works that show a WEFX allocation may not exist for two agents with general, additive valuations \cite{hajiaghayi2023wef}, and highlights that binary additive valuations may possibly be the best that one can hope for in terms of achieving WEFX (or even for WWEFX).
We first formally define WWEFX.
\begin{definition}[WWEFX]
    An allocation $\mathcal{A} = (A_1,\dots,A_n)$ is \emph{weak weighted envy-free up to any good (WWEFX)} if for all agents $i,j\in N$ and all $g \in A_j$, it holds that $\frac{v_i(A_i)}{w_i} \geq \frac{v_i(A_j \setminus \{g\})}{w_j}$ or $\frac{v_i(A_i \cup \{g\})}{w_i} \geq \frac{v_i(A_j)}{w_j}$.
\end{definition}

Then, our (negative) results are as follows.
\begin{proposition}\label{thm: WEFX, b,s}
    Under binary submodular valuations, a \emph{WWEFX} allocation may not exist, even when $n=2$.
\end{proposition}
   
\begin{proposition}\label{WEFX, restricted additive}
    Under restricted additive valuations, a \emph{WWEFX} allocation may not exist, even when $n=2$.
\end{proposition}

To end off this section, we show that a constant-factor approximation to WEFX when $n=2$ exists.
Formally, for $\alpha \in [0,1]$, we say that an allocation $\mathcal{A}$ is $\alpha$-WEFX if for all agents $i,j \in N$ and all $g \in A_j$,
    $\frac{v_i(A_i)}{w_i} \geq \alpha \cdot \frac{v_i(A_j \setminus \{g\})}{w_j}$.

\citet{suksompong2025weightedreview} raised as an open question the existence of such a constant $\alpha$, even for the case of $n=2$.
Here, we show the existence of a $\frac{1}{4}$-WEFX allocation, thereby improving upon the 
$\frac{w}{2\sqrt[3]{m}}$-WEFX result by \citet{hajiaghayi2023wef}, where $w$ is the larger weight of the two agents.
Our result is as follows.

\begin{theorem}\label{1/4-WEFX}
    When $n = 2$, there always exists a \emph{$\frac{1}{4}$-WEFX} allocation.
\end{theorem}

\section{EFX+} \label{sec:EFX+}
In the standard definition of EFX (see Definition \ref{def:EFX}), a good is dropped from the right-hand side (i.e., the `other agent's' bundle).
One could also consider a variant whereby the good is added on the left-hand side instead (i.e., the `evaluating agent's' bundle)---we call this \emph{EFX+}.
Under additive valuations, EFX+ is equivalent to EFX.
However, under general monotone valuations, EFX+ and EFX are incomparable (i.e., one does not necessarily imply the other).
Conceptual variants of this nature can be also found for envy-freeness properties studied in weighted fair division
\cite{chakraborty2021weighted,ChakrabortySeSu22,MontanariScSu24}.
It would thus be interesting to consider EFX+, which is defined as follows.

\begin{definition}[EFX+]
    An allocation $\mathcal{A} = (A_1,\dots,A_n)$ is \emph{EFX+} if for all agents $i,j \in N$ and all $g \in A_j$, $v_i(A_i \cup \{g\}) \geq v_i(A_j)$.
\end{definition}

We first give an example to illustrate the non-equivalence (and incomparability) of EFX and EFX+ under general monotone valuations. 
Consider an instance with $n=2$ agents and $m=3$ goods such that for both agents, $v(g_1) < v(g_2) < v(g_3) < v(g(\{g_2,g_3\}) < v(g(\{g_1,g_2\}) < v(g(\{g_1,g_3\})$. 
Then, the allocation that gives $g_1$ to an agent and $g_2,g_3$ to the other would satisfy EFX+ but not EFX; whereas the allocation that gives $g_3$ to an agent and $g_1,g_2$ to the other would satisfy EFX but not EFX+.

Next, we show that for two agents, EFX+ allocations exist, and that there is always a minimum of at least two EFX+ allocations.
To prove this, we employ the `cut-and-choose' protocol, where one agent partitions the set of goods into two `almost equal' bundles, and the other agent chooses their favourite bundle. The non-trivial part of this protocol involves showing that the first agent always has a valid cut. 
In a slight variation of \emph{leximin++} partial order \cite{plaut2020almost}, we introduce the notion of \emph{leximax}.
Our result is as follows.

\begin{theorem}\label{EFX+,n=2}
    When $n=2$, there are at least two \emph{EFX+} allocations.
\end{theorem}
We note that the use of \emph{leximax} can be used to show that EFX+ allocations also exists for any number of agents with identical valuations. While an EFX+ allocation always exists with two agents (or trivially when $m \leq n$), perhaps surprisingly (and in contrast to EFX), an EFX+ allocation may fail to exist with just three agents and four goods, as illustrated in the following result.
\begin{proposition}\label{EFX+,3 agents 4 goods}
    When $n=3$ and $m=4$, an \emph{EFX+} allocation may not exist.
\end{proposition}
\section{Conclusion}
In this work, we investigate the minimum number of EFX allocations in various restricted instances, and with respect to WEFX and EFX+. 
For EFX, we obtain tight bounds in the easy cases when $m \leq n+1$, and a lower bound of $n$ and upper bound of $n^2$ allocations when $m = n+2$. 
For the setting with $m = n+2$, we found instances where there are just $n$ allocations when $n \leq 4$ (see Remark 1 in the appendix). However, when $n = 5$, even after extensively generating random additive instances, our (empirical) upper bound on the number of EFX allocations is 14.
Thus, we present the following conjecture.
\begin{conjecture}
    With $n$ agents, there are at least $n$ EFX allocations.
\end{conjecture}

We also showed the existence and polynomial time computability of WEFX and PO allocations for any number of agents under binary additive valuations, and a $\frac{1}{4}$-approximation to WEFX for the case of $n=2$ agents, thereby resolving open problems in the area.

We also note that the negative results in the two variants considered highlights that EFX is surprisingly brittle. For instance, in Proposition \ref{WEFX, restricted additive}, we observe that even in the case of $n=2$ agents with slightly differing entitlements, and under restricted additive valuations, a WWEFX allocation may fail to exist. In Proposition \ref{EFX+,3 agents 4 goods}, by slightly tweaking the definition of EFX to EFX+, an EFX+ allocation may not exist with just $n=3$ agents and $m=4$ goods.

\bibliography{abb,aaai25}

\newpage

\appendix
\onecolumn

\newpage 
\begin{center}
\Large
\textbf{Appendix}
\end{center}
\section{Omitted Proofs from Section 3}
\subsection{Proof of \Cref{prop:bobw}}
    We consider the `cut-and-choose' protocol where an agent $i$ partitions $G$ into two sets $G'$ and $G\setminus G'$ such that $|v_i(G') - v_i(G\setminus G')|$ is minimized and the other agent $2$ choose their preferred bundle. 
    Now, consider the (possibly identical) allocations $(G', G\setminus G')$ and $(B, G\setminus B)$ where the former allocation is induced by agent $1$ partitioning the goods and agent $2$ selecting their favourite bundle and vice versa for the latter allocation.
    Clearly, both allocations are EFX. 
    Consider an outcome $\mathbf{o}$ that selects each of these allocations with an equal probability. 
    We note that $|v_1(G') - v_1(G\setminus G')| \leq v_1(B) - v_1(G\setminus B) \implies  v_1(G\setminus G') - v_1(G')  \leq v_1(B) - v_1(G\setminus B) \iff v_1(G') + v_1(B) \geq v_1(G\setminus G') + v_1(G\setminus B) $. Thus, agent~$1$ does not envy agent~$2$ ex-ante. We can conduct a symmetric analysis for agent~$2$. 
    Hence, $\mathbf{o}$ is ex-ante EF and ex-post EFX .

\subsection{Proof of \Cref{small is hard}}
We first provide the main idea of the proof.
Suppose we know some $f(n)$ for which EFX existence is guaranteed for $m = n + \lfloor f(n) \rfloor$. Given some fixed $n,m$, if $m = n + \lfloor f(n) \rfloor$ then we are done; if $m < n + \lfloor f(n) \rfloor$ then we reduce this to the previous case by creating $n + \lfloor f(n) \rfloor - m$ many dummy goods valued at $0$ by every agent; if $m > n + \lfloor f(n) \rfloor$ then we give a reduction that adds $k$ agents and $k$ goods, and since $f \in \omega(1)$ there is always a natural $k$ such that $m + k \leq n + k + \lfloor f(n+k) \rfloor$ (or equivalently, $m \leq n + \lfloor f(n+k) \rfloor$ ).

We proceed to prove the result. 
Suppose an EFX allocation always exists when $m = n + f(n)$, where $n = \omega(1)$. For any monotonic function $f$, if $f = \omega(1)$, then for every $y$, there is a value of $x$ such that $f(x) \geq y$.

It is easy to see that the case where $m < n + f(n)$ can be easily reduced to the previous case by creating $n + f(n) - m$ many dummy goods valued at $0$ by every agent. Thus, we assume there is an instance $\mathcal{I}$ with $n$ agents (in $X$) and $m$ goods (in $G$) such that $m > n + f(n)$. 
Let $k$ be the smallest value such that $f(k+n) > m - n$. 
Consider another instance $\mathcal{I'}$ constructed from $\mathcal{I}$, but we introduce $k$ additional agents (let this additional set be $X'$) and $k$ additional goods (let this additional set be $G'$). 
In this new instance $\mathcal{I'}$, there are $k+n$ agents and $k+m < k+n+f(k+n)$ goods. Suppose agent $y \in X$. In this instance $\mathcal{I'}$, we construct the valuations function as follows.
\begin{equation*}
    v'_x(g) =
    \begin{cases}
        \infty & \text{if  } \, x \in  X' \cup \{y\} \,\,\land\,\, g\in G'\\
                0 & \text{if  } \, x \notin  X' \cup \{y\} \,\,\land\,\, g\in G'\\
        v_{y}(g) & \text{if  } \, x\in  X' \cup \{y\}\,\,\land\,\, g\in G\\
                v_x(g) & \text{if  } \, x \notin  X' \cup \{y\} \,\,\land\,\, g\in G\\
    \end{cases}
\end{equation*}

We show that an EFX allocation for instance $\mathcal{I'}$ can be used to construct an EFX allocation for $\mathcal{I}$. 
In other words, there is an EFX allocation for $\mathcal{I'}$ only if there is an EFX allocation for $\mathcal{I}$.

First, observe that in any EFX allocation for $\mathcal{I'}$, any agent that is allocated a good in $G'$ have to get only one good. 
As $|G'| < |X'\cup \{y\}|$, there must be an agent $ x \in X'\cup \{y\} $ that is not allocated a good in $G'$. 
If any agent is allocated a good in $G'$ along with an additional good, $x$ will envy that agent even after dropping the least-valued good in that agent's bundle.

Next, observe if there exists an EFX allocation for $\mathcal{I'}$, then there exists an EFX allocation where all agents in $X'$ receives all the goods in $G'$ and nothing else. 
Suppose that agent $y$ has a good in $G'$ and hence there is an agent $x \in X'$ that does not receive a good in $G'$. 
As both agents have identical valuations, swapping the bundle of goods they receive in an EFX allocation still yields an EFX allocation. 
Then, suppose that agent $x \neq y, x \in X$ receives a good in $G'$ and we have an agent $z \in X'$ does not. We know that agent $x$ is only allocated the good in $G'$, and hence her bundle value is $0$; thus, any swap cannot not decrease her valuation of her bundle. Meanwhile, if agent $z$ receives a good in $G'$ via a swap, her valuations of her bundle will increase. As such, swapping the bundle of goods allocated to $x$ and $z$ will still yields an EFX allocation. 

Hence, through a series of swaps, each agent in $X'$ will obtain a good in $G'$ and by our observation they must only be allocated that good. 
In other words, all agents in $X$ receive only goods in $G$ and any good in $G$ is only allocated to agents in $X$. 
Such an allocation can thus be translated to the original instance $\mathcal{I}$, giving us an EFX allocation for $\mathcal{I}$. 

\subsection{Proof of Lemma \ref{m,n+2,lemma 1}}
To find an EFX allocation whereby an particular agent has at least two goods, we first label such an agent as $n$ (without loss of generality) to position the agent last in the initial ordering. 
        We split our analysis into two cases, depending on the existence of a non-envied agent in the partial allocation $\mathcal{A}$ at State 1.
        Suppose there exists a non-envied agent $j \in N$.
        If $n = j$, then we can allocate all goods in $P$ to agent $n$, and if $n \neq j$, we allocate good $l$ to agent $n$, and good $s$ to agent $j$ 
        In both cases, we get an EFX allocation satisfying the criteria in the lemma: in the first case $n$ receives three items, whereas in the second case $j$ receives her most valued item in $A_n \cup A_j$.  
        Now, suppose there does not exist a non-envied agent, i.e., all agents are envied in $\mathcal{A}$. 
        This means that there is an agent $i$ that prefers $l$ to her own bundle $A_i$. 
        We find a path from $i$ to $n$ in the envy cycle with the following algorithm.

        \begin{algorithm}
        \caption{Returns an envy cycle where agent $n$ is unenvied.}
        \begin{algorithmic}[1] 
            \newcommand{\Step}[2]{\STATE #1 \hfill$\triangleright$ \textit{#2}}
            \STATE \textbf{Input:} set of agents $N$, set of goods $G$, valuation profile $\mathbf{v} = (v_1,\dots,v_n)$
            \STATE Initialize $k := i$
            \STATE Initialize $\texttt{Path} \leftarrow [i]$
            \WHILE{$k \neq n$}
                \STATE Let $j$ with the smallest indexed agent such that $j$ envies $k$
                \STATE $\texttt{Path} \leftarrow \texttt{Path} + [j]$
                \STATE $k \leftarrow j$
            \ENDWHILE
            \STATE \textbf{return} \texttt{Path}
        \end{algorithmic}
        \end{algorithm}
        
        We note that for any agent $i \in [n-1]$, if $i$ is envied by agent $j$, then $i < j$. Hence, if $k < n$, there is a $j > k$ such that $j$ envies $k$. As $k$ increases in each iteration of the \textbf{while} loop, the algorithm terminates and provides a path from $i$ to $n$ such that if $(k,j)$ is in the path $j$ envies $k$. 
        
        Let agent $x$ be the penultimate agent in the path. In other words, $(x,n)$ is in the path. As there are no other agents with a smaller index than $n$ that envies agent $x$, $n$ is the only agent that envies $x$. After performing envy cycle elimination according to the path found in algorithm 2, $n$ will receive the good originally in $A_x$. As agents can not decrease their utility from envy cycle elimination, no agent will prefer the good originally in $A_x$ over their own allocated goods. Hence, agent $n$ is non-envied and can receive the good $s$, yielding an EFX allocation where $n$ has two goods. Also note that agent $n$ prefers the good originally in $A_x$ to any of the goods in $P$.
        
        In all the cases, it is easy to see that all agents receives at least one good in the EFX allocation constructed.

\subsection{Proof of Lemma \ref{lem:joingraph|E|}}

We first note that each edge $(i,j) \in E$ in the join-graph corresponds to an EFX allocation whereby agents $i,j$ receives more than one good and all agents in $N \setminus \{i,j\}$ receives exactly one good. 
    We then show that each edge corresponds to a different EFX allocation, and therefore there are at least $|E|$ EFX allocations. 
    Take any two distinct edges $(i,j) , (i',j') \in E$.
    We split our analysis into two cases.
    \begin{description}
        \item[Case 1: $i = j'$ and $j = i'$.] 
        First note that $i \neq j$ as $(i,j)$ and $(i',j')$ are distinct edges.
        Suppose both $(i,j) $ and $ (j,i)$ correspond to the same EFX allocation $\mathcal{A}$.
        Let $A_i = \{x_1,x_2\}$ and $A_j = \{y_1,y_2\}$.
        As $\mathcal{A}$ is EFX, all agents in $N\setminus\{i,j\}$ prefers the goods the receive to any good in $\{x_1,x_2,y_1,y_2\}$. Without loss of generality, let $v_i(x_1) \geq v_i(x_2)$ and $v_j(y_1) \geq v_j(y_2)$. 

        Consider the allocation $A'$ where agents in $N \setminus \{i,j\}$ receives the same good as they did in $\mathcal{A}$, $i$ receives $\{x_1,y_2\}$ and $j$ receives $\{y_1,x_2\}$. 
        As $\mathcal{A}$ corresponds to both $(i,j)$ and $(j,i)$, $v_i(x_1) \geq v_i(y_1), v_i(x_1) \geq v_i(y_2)$ and $v_j(y_1) \geq v_j(x_1), v_j(y_1) \geq v_j(y_2)$. 
        $A'$ is distinct from $\mathcal{A}$ and is an EFX allocation that can correspond to both $(i,j)$ and $(j,i)$. 
        Thus $(i,j)$ and $(j,i)$ corresponds to two different EFX allocations.

        \item[Case 2: $i \neq j'$ or $i \neq j'$.]
        Suppose, without loss of generality, that $i \neq j'$. 
        As $(i,j)$ and $(i',j')$ are distinct edges, either $i \neq i'$ or $j \neq j'$. 
        In the first case $i \notin \{i',j'\}$ and in the second case $j' \notin \{i,j\}$. In the first case, in the EFX allocation corresponding to $(i,j)$, $i$ receives more than one good but in the EFX allocation corresponding to $(i',j')$, $i$ receives exactly one good. 
        In the second case, in the EFX allocation corresponding to $(i',j')$, $j'$ receives more than one good but in the EFX allocation corresponding to $(i,j)$, $j'$ receives exactly one good. 
        Hence, the EFX allocations corresponding to $(i,j)$ and $(i',j')$ are different. 
    \end{description}

\subsection{Proof of Theorem \ref{thm: counting: identical}}

Next, we prove the upper bound.
Construct an instance where agents have valuations over goods defined as follows: $v(g_j) = 1$ if $j\leq n- 1$ and $v(g_j) = 0$ otherwise. 
For all allocations, there exist an agent that receives no good from the set  $\{g_1, \dots, g_{n-1}\}$. 
Such an agent will have a value of $0$ for their own bundle (from our definition of the valuation function). 
Hence, in every EFX allocation, if an agent $i \in [n]$ receives more than one good, none of the goods can be in the set $\{g_1, \dots, g_{n-1}\}$. Otherwise, this agent will be envied (after dropping the least-valued good) by an agent that receives no good from the set  $\{g_1, \dots, g_{n-1}\}$. 
This means that in every EFX allocation, $n-1$ agents will each receive exactly one good from the set $\{g_1, \dots, g_{n-1}\}$, and one agent will receive all remaining goods. This yields $n!$ EFX allocations, corresponding to the number of permutations of the bundles. 

\section{Omitted Proofs from Section 4}
\subsection{Proof of Theorem \ref{thm:WEFX,n =2, PO + B,A}}
We first show the existence of WEFX and PO allocations with a welfare function.
Recall the \emph{leximin++} comparison operator introduced by \cite{plaut2020almost} that showed that a \emph{leximin++} optimal allocation is EFX under identical valuations. We introduce $\prec_{lex++}$ which slightly modifies the \emph{leximin++} partial order to account for the weights.

Formally, for any allocation $\mathcal{A} = (A_1,\dots,A_n)$ and agent $i \in N$, we denote $u_i(A_i) = v_i(A_i)/w_i$.
Also let $\pi_\mathcal{A}: N \rightarrow N$ be a permutation of the agents such that if $i<j$, then $u_{\pi_\mathcal{A}(i)}(\mathcal{A}) \leq u_{\pi_\mathcal{A}(j)}(\mathcal{A})$. Now, we define the partial order $\prec_{lex++}$ such that for any two allocations $\mathcal{A}$ and $\mathcal{B}$, we say that $\mathcal{A} \prec_{lex++} \mathcal{B}$ if either
\begin{enumerate}[(i)]
    \item there exists an agent $k \in N$ such that $u_{\pi_\mathcal{A}(k)}(\mathcal{A}) < u_{\pi_B(k)}(\mathcal{B})$ and for all $i < k$, $u_{\pi_\mathcal{A}(i)}(\mathcal{A}) = u_{\pi_B(i)}(\mathcal{B})$, or
    \item for all $i \in N$, $u_{\pi_A(i)}(\mathcal{A}) = u_{\pi_B(i)}(\mathcal{B})$, and there exists $k$ such that $|A_{\pi_\mathcal{A}(k)}|< |B_{\pi_\mathcal{B}(k)}|$ and for all $i < k$, $|A_{\pi_\mathcal{A}(i)}|= |B_{\pi_\mathcal{B}(i)}|$.
\end{enumerate} 
We say that an allocation $\mathcal{A}$ is $\prec_{lex++}$-optimal if it is a maximal element (over all allocations) according to $\prec_{lex++}$. Now, let $\mathcal{A}$ be a $\prec_{lex++}$-optimal allocation. As $\mathcal{A}$ is clearly PO, we focus on proving $\mathcal{A}$ is WEFX. 

First note that for all agents $j \in N$, if there exists a good $g \in A_j$ such that $v_j(g) = 0$, then as $\mathcal{A}$ is PO, all other agents must also have a value $0$ for good $g$. 
This also means that for every other agent $i \in N \setminus \{j\}$, 
\begin{equation} \label{eqn:wefx_po_1}
    v_i(A_j) \leq v_j(A_j).
\end{equation}
Suppose for a contradiction that $\mathcal{A}$ is not WEFX, i.e., there exists agents $i,j \in N$ such that
\begin{equation} \label{eqn:wefx_po_2}
    \frac{v_i(A_i)}{w_i} < \frac{v_i(A_j \setminus \{g\})}{w_j}
\end{equation}
for some $g \in A_j$.
We then split our analysis into two cases.
\begin{description}
    \item[Case 1: $\exists g \in A_j$ such that $v_j(A_j) = 0$.] 
    Recall that all agents have value $0$ for good $g$. 
    Consider the allocation $\mathcal{A}'$ such that $A'_i = A_i \cup \{g\}$, $A'_j = A_j \setminus \{g\}$, and $A'_k = A_k$ for all $k \in N \setminus\{i,j\}$. 
    Then we have that 
    \begin{equation*}
        \frac{v_j(A'_j)}{w_j} = \frac{v_j(A_j\setminus \{g\})}{w_j} = \frac{v_j(A_j)}{w_j} \geq \frac{v_i(A_j)}{w_j} > \frac{v_i(A_i)}{w_i},
    \end{equation*}
    where the inequalities follows from (\ref{eqn:wefx_po_1}) and (\ref{eqn:wefx_po_2}) respectively.
    Furthermore, $ \frac{v_i(A'_i)}{w_i} =  \frac{v_i(A_i)}{w_i}$ with $|A'_i| > |A_i|$, and thus $\mathcal{A} \prec_{lex++} \mathcal{A}'$, contradicting our choice of $\mathcal{A}$.
    \item[Case 2: $\forall g \in A_j$,  $v_j(A_j) = 1$.]
    From (\ref{eqn:wefx_po_2}), we get that 
    \begin{equation*}
        \frac{v_i(A_i)}{w_i} < \frac{v_i(A_j \setminus \{g\})}{w_j} \leq \frac{|A_j| - 1}{w_j}.
    \end{equation*}
    Furthermore, as $\frac{v_i(A_i)}{w_i} < \frac{v_i(A_j)}{w_j}$, then $v_i(A_j) > 0$ and there exists a good $g^* \in A_j$ such that $v_i(g^*) = 1$. 
    Consider the allocation $\mathcal{A}'$ such that $A'_i = A_i \cup \{g^*\}$, $A'_j = A_j \setminus \{g^*\}$ and $A'_k = A_k$ for all $k \in N \setminus\{i,j\}$. 
    Then we have that 
    \begin{equation*}
    \frac{v_j(A'_j)}{w_j} = \frac{v_j(A_j\setminus \{g^*\})}{w_j} = \frac{|A_j| -1}{w_j} \geq \frac{v_i(A_j)}{w_j} > \frac{v_i(A_i)}{w_i}.    
    \end{equation*}
    Furthermore, $\frac{v_i(A'_i)}{w_i} = \frac{v_i(A_i \cup g^*)}{w_i} = \frac{v_i(A_i) + 1}{w_i} > \frac{v_i(A_i)}{w_i}$, and thus $\mathcal{A} \prec_{lex++} \mathcal{A}'$, contradicting our choice of $\mathcal{A}$.
\end{description}

Next, we show that there exists an algorithm that returns a WEFX and PO allocation in polynomial-time.

We first define what it means for an agent to be \emph{WEFX envied} for an allocation $\mathcal{A}$. 
\begin{definition}
    An agent $i$ is \emph{WEFX envied} by an agent $j$ if and only if $\frac{v_j(A_j)}{w_j} <\frac{v_j(A_i \setminus \{g\})}{w_i} $ for some $g \in A_i$.
\end{definition}

Note that if for all agents $i,j \in N$, $i$ is not WEFX envied by $j$ then the allocation is WEFX.
We now state a necessary and sufficient condition for an allocation $\mathcal{A}$ to be \emph{PO} under binary additive valuations. An allocation $\mathcal{A}$ is \emph{PO} if and only if
\begin{equation} \label{eqn:wefx_po1}
\forall i,j \in N, \forall g \in A_i,  \quad v_i(g) = 0 \implies  v_j(g) = 0
\end{equation}
In other words, if a good $g$ is not valued at $0$ by all agents, then $g$ must be given to an agent with a positive value for $g$. 
If an allocation $\mathcal{A}$ satisfies (\ref{eqn:wefx_po1}), then it also satisfies the following 
\begin{equation} \label{eqn:wefx_po2}
\forall i,j \in N, \quad v_i(A_j) \leq v_j(A_j)
\end{equation}

We first show that it is sufficient to consider the case where all goods are not valued at $0$ by all agents, with the following lemma.
\begin{lemma}\label{lemma:wefx_po_1}
    Let $G_{0} \subseteq G$ be the set of goods such that all agents have $0$ value all goods in $G_0$. 
    If there is a partial allocation $\mathcal{A}$ over $G \setminus G_0$ that is \emph{WEFX} and satisfies (\ref{eqn:wefx_po1}), then there is a complete allocation over $G$ that is \emph{WEFX} and \emph{PO}, and can be found in polynomial-time.
\end{lemma}
\begin{proof}
    Let $u_i(\mathcal{A}) = v_i(A_i)/w_i$. Next, identify the agent $i$ with the minimum $u_i(\mathcal{A})$. 
    Give all goods in $G_0$ to $i$ and let this complete allocation be $\mathcal{B}$. 
    We note that $\mathcal{B}$ can be found in polynomial time. 
    As $\mathcal{B}$ also fulfils (\ref{eqn:wefx_po1}), $\mathcal{B}$ is PO. 
    Furthermore, for all $j \in N$, 
    \begin{equation*}
        v_j(B_i)/w_i \leq v_i(B_i)/w_i = u_i(\mathcal{A}) \leq u_j(\mathcal{A}) = v_j(A_j)/w_j = v_j(B_j)/w_j
    \end{equation*}
    where the leftmost inequality follows from (\ref{eqn:wefx_po2})) and the middle inequality follows by the choice of agent $i$.
    Thus, $i$ is not WEFX envied by any agent and $\mathcal{B}$ is WEFX.
\end{proof}
Lemma \ref{lemma:wefx_po_1} allows us to ignore zero-valued goods for all agents. 
Thus, for the rest of the proof, we assume that each good is valued by some agent. 
In other words, for an allocation $\mathcal{A}$ to satisfy (\ref{eqn:wefx_po2}), then

\begin{equation} \label{eqn:wefx_po3}
\forall i \in N, v_i(A_i) = |A_i|
\end{equation}

Construct the following bipartite graph $G_{match} = (N, G , E)$ where $(i,j) \in E$ if and only if $v_i(g_j) = 1$. Let $\mathcal{M}$ be the maximal matching of $G_{match}$. Let $N_0$ be the set of agent $i$ such that for all $j \in [m], (i,j) \notin \mathcal{M}$. In other words, if we view $\mathcal{M}$ as a partial allocation, $N_0$ is the set of the agents that receives no goods. We now show that we can `ignore' the agents in $N_0$ if the following condition holds. 

\begin{lemma}\label{lemma:wefx_po_2}
    Suppose each good in $G$ is valued by some agent. Let $\mathcal{I}'$ be the instance over the agents in $N\setminus N_0$ and all the goods $G$. If there is a \emph{WEFX} allocation $\mathcal{A}$ that satisfies (\ref{eqn:wefx_po1}) and for all $i \in N$, $v_i(\mathcal{A}_i) \geq 1$, then $\mathcal{A}$  is \emph{WEFX} and satisfies (\ref{eqn:wefx_po1}) for the original instance with set of agents $N$.
\end{lemma}
\begin{proof}
    Clearly, $\mathcal{A}$ still satisfies (\ref{eqn:wefx_po1}) even if we include agents in $N_0$. 
    Thus, it suffices to show that $\mathcal{A}$ is WEFX.
    As all agents in $N_0$ receives no goods, they would not be WEFX envied. 
    Furthermore, if $i,j \in N \setminus N_0$, as $\mathcal{A}$ is WEFX and $i$ is not WEFX envied by $j$. For our final case where $i \in N \setminus N_0, j \in N_0$, if $v_j(A_i) > 0$, then $|A_i| = 1$ and $i$ is not WEFX envied by $j$.

    Suppose there was an agent $i \in N \setminus N_0, j \in N_0$ such that $v_j(A_i) > 0$ and $|A_i|$. 
    Then there is a good $g$ in $A_i$ such that $v_j(g) = 1$. 
    By (\ref{eqn:wefx_po3}), $v_i(A_i \setminus \{g\}) > 0$. Thus, by allocating $g$ to agent $j$ instead of $i$, as for all $i \in N$, $v_i(\mathcal{A}_i) \geq 1$, we have an allocation such that $|N \setminus N_0| + 1$ agents receive at least $1$ good they positively value. This implies that in $G_{match}$, and there is a matching of size $|N \setminus N_0| + 1$, thereby contradicting the maximality of $\mathcal{M}$.
\end{proof}

Lemma \ref{lemma:wefx_po_2} allows us to essentially `ignore' the agents in $N_0$ and allocate all the goods solely to agents in $N \setminus N_0$. 
Together with lemma \ref{lemma:wefx_po_1}, we can now restrict ourselves to the case where $G$ contains no zero-value goods and there is an allocation $\mathcal{A}'$ such that for all $i \in N$, $v_i(A'_i) \geq 1$. We then conclude our proof with the following lemma.

\begin{lemma}
    Suppose $G$ contains no zero-value goods and there is an allocation $\mathcal{A}'$ such that for all $i \in N$, $v_i(A'_i) \geq 1$. Then, there is a \emph{WEFX} allocation $\mathcal{A}$ that satisfies  (\ref{eqn:wefx_po1}) and for all $i \in N$, $v_i(\mathcal{A}_i) \geq 1$. Moreover, $\mathcal{A}$ can be computed in polynomial time.
\end{lemma}

Let $V = \{\frac{j}{w_i} \mid i \in N, j \in [m+1]\}$. Then, pick an $\varepsilon$ such that for all $i \in N$, $\lceil \varepsilon w_i \rceil = 1$ and $\varepsilon$ is smaller than all elements in $V$. Let $V^+ = V \cup \{g\}$. For $x \in V^+$, let $succ(x)$ be the smallest value $y \in V^+$ such that $succ(x) < y$. 

For map $M$, let $Add(M, (k,v))$ be the operation that adds the key-value pair $(k,v)$ to $M$, let $contain(M,k)$ return True if the key $k$ is in $M$ and false otherwise, and $M_k$ be the value associated with the key $k$ in $M$.

Then, our algorithm is as follows (Algorithm \ref{alg:wefx1}).
\begin{algorithm}
\caption{Returns a WEFX and PO allocation.}
\label{alg:wefx1}
\begin{algorithmic}[1] 
\STATE \textbf{Input:} Instance $I$, $\varepsilon$ and $V^+$
\STATE Initialize an empty hashmap, $M \leftarrow \emptyset $
\STATE Initialize the value of $k$: $k \leftarrow \varepsilon$
\STATE \textbf{Loop Invariant:}
\begin{enumerate}
    \item There is an allocation $\mathcal{A}$ such that for all $i \in N$, $v_i(A) \geq M_i \cdot w_i$ if $contain(M,i) = $ True, and $v_i(A_i) \geq k \cdot w_i$ if $contain(M,i) = $ False 
    \item Suppose $(i, v) \in M$, let $N_{s}$ be the set of agents $j \neq i$ such that $contain(M,j) = $ True and $M_j \leq v$. 
    Then, there is no allocation $\mathcal{A}'$ such that for all $j \in N_{s}$, $v_j(A'_j) \geq M_j \cdot w_j$ and for all agents in $N \setminus N_{s} > v \cdot w_j$.
 \end{enumerate}

\WHILE{there exists an agent $i \in N$ such that $contain(M,i) =$ False}
\STATE  $M^+ \leftarrow \texttt{FindMinimal}(\mathcal{I}, M, V^+, k)$ \quad (refer to the subroutine Algorithm \ref{alg:wefx2})
\FOR{$i \in M^+$}
    \STATE $Add(M, (i,k))$
\ENDFOR
\STATE $k \leftarrow succ(k)$
\ENDWHILE
\RETURN $\mathcal{A}$ (derived from $M$)
\end{algorithmic}
\end{algorithm}

\begin{algorithm}
\caption{\texttt{FindMinimal} Subroutine}
\label{alg:wefx2}
\begin{algorithmic}[1] 
\STATE \textbf{Input:} Instance $\mathcal{I}$, $M$, $V^+$ and $k$
\STATE \textbf{Output:} A set of agents $S$ such that 
\begin{enumerate}
    \item There is an allocation $\mathcal{A}$ such that for all agents $j \in N$, $v_j(A_j) \geq M_j \cdot w_j$ if $contain(M,j) = true$,  $v_j(A_j) \geq k \cdot w_j$ if $j \in S$ and   $v_j(A_j) > k \cdot w_j$ otherwise
    \item  For all agents $i \in S$, there is no allocation $\mathcal{A}$ such that $v_i(A_i) > k \cdot w_i$ and for all agents $j \in N \setminus \{i\}$, $v_j(A_j) \geq M_j \cdot w_j$ if $contain(M,j) = True$,  $v_j(A_j) \geq k \cdot w_j$ if $j \in S$ and   $v_j(A_j) > k \cdot w_j$ otherwise
\end{enumerate}
\STATE \textbf{Precondition:}
\begin{enumerate}
    \item There is an allocation $\mathcal{A}$ such that for all $i \in N$, if $contain(M,i) = true$, then $v_i(A) \geq M_i \cdot w_i$ and if $contain(M,i) = false$, then $v_i(A) \geq k \cdot w_i$
    \item Suppose $(i, v) \in M$, let $N_{s}$ be the set of agents $j \neq i$ such that $Contain(M,j) = true$ and $M_j \leq v$. Then there is no allocation $\mathcal{A}$ such that for all $j \in N_{s}$, $v_j(A) \geq M_j \cdot w_j$ and for all agents in $N \setminus N_{s} > v \cdot w_j$.\end{enumerate}
\STATE Consider the bipartite graph and $G_{s^*} = (N_{aug}, G_{aug}, E)$ for some constant $s^*$. 
\FOR{$i \in N$}
    \STATE Introduce to $N_{aug}$ the nodes $\{i_1, \dots , i_{\tau_i}\}$ where $\tau_i = M_i \cdot w_i$ if $contain(M,i) = $ True and  $\tau_i = succ(k) \cdot w_i$ if $contain(M,i) = $ False.
\ENDFOR
\FOR{$g \in G$}
\STATE Introduce a node in $G_{aug}$, and nodes $S = \{s_1, \dots, s_{size}\}$.
\ENDFOR
\FOR{$i_j \in N_{aug}$ and $g \in G$}
    \IF{$v_i(g) \geq 1$}
        \STATE Add $(i_j, g)$ to $E$ 
    \ENDIF
\ENDFOR
\FOR{$s \in S$}
    \IF{$contain(M,i) = $ False and $j  = \tau_i$ and $\lceil k \cdot w_i \rceil <\lceil succ(k) \cdot w_i \rceil$}
    \STATE Add $(i_j, g)$ to $E$
    \ENDIF
\ENDFOR
\STATE Find the smallest value of $s^*$ such that there is a matching $M$ of size $|N_{aug}|$.  Then, $S$ contains the agents $i$ such that $\exists s \in S, (i_{\tau_i}, s) \in M$
\RETURN $S$
\end{algorithmic}
\end{algorithm}
We first show that Algorithm \ref{alg:wefx1} terminates in polynomial time if Algorithm  \ref{alg:wefx2} terminates in polynomial time; and that the loop invariant is initially True, and will remain True if the output of Algorithm~\ref{alg:wefx2} fulfils criteria set out in Algorithm~\ref{alg:wefx2}. 
Intuitively, Algorithm~\ref{alg:wefx1} will identify how many goods to allocate to each agent, and the matching tells us which allocation achieves this.

We then show that Algorithm \ref{alg:wefx2} terminates in polynomial time and the criteria set out in Algorithm \ref{alg:wefx2} is satisfied. 
Finally, we show that how to derive a WEFX allocation $\mathcal{A}$ that satisfies (\ref{eqn:wefx_po1}) and that for all $i \in N$, $v_i(\mathcal{A}_i) \geq 1$.

First, we note that $|V^+| = \mathcal{O}(nm)$. 
Furthermore, there is a $k \in V^+$ such that $M$ contains all the agents. 
Pick an agent $i$ and let $k = \frac{m+1}{w_i}$. Observe that $k \in V^+$
Then, if $contain(M,i) = $ False, by our first loop invariant, there is an allocation where $i$ receives $m+1$ goods, which gives us a contradiction. 
Thus, the \textbf{while} loop must run at most $|V^+|$ times. Next, we initially have that $k = \varepsilon$. Since $\varepsilon \cdot w_i \leq 1$ and there is an allocation $\mathcal{A}'$ such that for all $i \in N, v_i(A'_i) \geq 1$ (by premise of the lemma), we have that \emph{loop invariant 1} is true. 
As $M$ is initially empty, \emph{loop invariant 2} is trivially true. 
Condition 1 of the output of Algorithm \ref{alg:wefx2} guarantees that \emph{loop invariant 1} remains true and condition 2 of the output of Algorithm \ref{alg:wefx2} guarantees that  \emph{loop invariant 2} remain true.

Next, we note that by the construction of $V^+$, for all $i \in N$, $\lceil k \cdot w_i \rceil$ and $\lceil succ(c) \cdot w_i \rceil$ differ by at most $1$ good. 
Thus, a matching of size $|N_{aug}|$ in $G_{size}$ corresponds to a partial allocation where for all $i \in N$, $v_i(A_i) = M_i \cdot w_i$ if $contain(M,i) = true$,  
$v_i(A_i) = k \cdot w_i$ if $\exists s \in S (i_{\tau_i}) \in S$ and $v_i(A_i) = succ(k) \cdot w_i$ if $\exists s \in S (i_{\tau_i}) \in S$ otherwise. 
Hence, as $S$ corresponds to matching of size $|N_{aug}|$ whereby, \emph{condition 1 of the output of algorithm \ref{alg:wefx2}} is satisfied. Furthermore, as $s^*$ is the minimum possible value, for all agents $i \in S$, there is no allocation $\mathcal{A}$ such that $v_i(A_i) > k \cdot w_i$ and for all agents $j \in N \setminus \{i\}$, $v_j(A_j) \geq M_j \cdot w_j$ if $contain(M,j) = $ True,  $v_j(A_j) \geq k \cdot w_j$ if $j \in S$ and   $v_j(A_j) > k \cdot w_j$ otherwise. 
Otherwise, a matching $M$ derived from $A_j$ will have size $|N_{aug}|$ with a lower value of $s^*$. 
\emph{Condition 2 of the output of algorithm \ref{alg:wefx2}}. 
By $\emph{precondition 1}$, there is an allocation $\mathcal{A}$ such that $v_i(A_i) = M_i \cdot w_i$ if $contain(M,i) = $ True,  $v_i(A_i) = k \cdot w_i$ otherwise and thus if $s^* \geq |N|$, there is a matching of size $|N_{aug}|$. Thus, the smallest value of $s^*$ is $\mathcal{O}(n)$ and as we can find maximal matching in polynomial time, and Algorithm \ref{alg:wefx2} runs in polynomial time.  

Lastly, after Algorithm \ref{alg:wefx1} terminates, these are the post-conditions.
\begin{enumerate}
    \item There is an allocation $\mathcal{A}$ such that for all $i \in N$, $v_i(A) \geq M_i \cdot w_i$
    \item For $(i, v) \in M$, let $N_{s}$ be the set of agents $j \neq i$ such that $M_j \leq v$. Then there is no allocation $\mathcal{A}$ such that for all $j \in N_{s}$, $v_j(A_j) \geq M_j \cdot w_j$ and for all agents in $N \setminus N_{s} > v \cdot w_j$.
 \end{enumerate}

 Post-condition 2 has the following implication:
    For $(i, v) \in M$, then there is no allocation $\mathcal{A}$ such that for all $j \in N \setminus \{i\}$, $v_j(A_j) \geq M_j \cdot w_j$ and $v_i(A_i) > M_i \cdot w_i$.

We first note that for all agents $M_i \cdot w_i$ is a whole number. Otherwise, if there is an allocation such that for all $j \in N$ $v_j(A) \geq M_j \cdot w_j$, then there is an allocation such that for all $j \in N\setminus \{i\}$ $v_j(A) \geq M_j \cdot w_j$ and $v_j(A) > M_j \cdot w_j$, thus contradicting implication of post-condition 2.

Now, we can construct $\mathcal{A}$ from $M$ again through a matching. $G_A = (N_{aug}, G, E)$. For each $i \in N$, introduce to $N_{aug}$ the nodes $\{i_1, \dots , i_{M_i \cdot w_i}\}$. For $i_j \in N_{aug}$, for $g \in G$ , then $(i_j, g) \in E$. Then, there is a maximal matching in $G_{A}$ (by post-condition 1) that corresponds to an allocation $\mathcal{A}$ such that $v_i(A_i) = M_i \cdot w_i$.  

We note that $\mathcal{A}$ is a complete allocation. Suppose, there exists $g \in G$ that was not allocated in $\mathcal{A}$. As there are no zero-valued goods, there is an agent $j$ such that $v_j(g) = 1$. By giving $g$ to agent $j$, we can derive an allocation that contradicts the implication of post-condition 2.

As in $G_A$ there is only a edge between $i_j$ and $g$ if $v_i(g) = 1$, then allocation $\mathcal{A}$ which comes from a matching fulfils condition \ref{eqn:wefx_po_1}. Furthermore, as for all $i\in N$, $M_i \geq \varepsilon > 0$, $v_i(A_i) = M_i \cdot w_i \geq 1$. 
Suppose $\mathcal{A}$ is not WEFX. Then there exists a agent $i,j$ and a $g \in A_j$ such that $\frac{v_i(A_i)}{w_i} < \frac{v_i(A_j \setminus \{g\}}{w_j} \leq \frac{|A_j| - 1}{w_j} < \frac{|A_j|}{w_j}$. Thus, $M_i < M_j$. Furthermore, by giving $g$ to agent $j$, we have an allocation whereby $\frac{v_i(A_i \cup \{g\})}{w_i} >  \frac{v_i(A_i)}{w_i} \geq M_i$ and $\frac{v_j(A_j \setminus \{g\})}{w_i} >  \frac{|A_j| - 1}{w_j} \geq M_j$ (where the second last inequality is due to the fact that for all $g \in A_j$, $v_j(A_j)  = 1$). Thus, consider $N_{s}$ for the agent $i$. $j \notin N_{s}$. Then by giving $g$ to agent $i$, there is an allocation $\mathcal{A}$ such that for all $j \in N_{s}$, $v_j(A_j) \geq M_j \cdot w_j$ and for all agents in $N \setminus N_{s} > v \cdot w_j$. This contradicts the implication of post-condition 2.

\subsection{Proof of Proposition \ref{1 allocation: WEFX}}
    For the case when $m \leq 2$, we can follow the same argument as in Theorem \ref{thm: counting: n<=m}.
    Thus, we focus on the case when $m \geq 3$.
    The lower bound is guaranteed by Theorem \ref{thm:WEFX,n =2, PO + B,A}. 
    For the upper bound, let the weights be $w_1 = 1 - \epsilon$ and $w_2 = \epsilon$ and valuation functions be defined as follows: $v_1(g_j) = 1$ for all $j \in [m]$, and $v_2(g_j) = 1$ if $j = 1$, $v_2(g_j) = 0$ otherwise. 
    Consider some allocation $\mathcal{A} = (A_1,A_2)$.
    If $|A_2| \geq 2$, then $v_1(A_2 \setminus \{g\}) \geq 1$ and agent~$1$ envies agent~$2$ even after dropping the least-valued good (in agent~$1$'s view) from agent~$2$'s bundle, and hence $\mathcal{A}$ does not satisfy WEFX.
    If $|A_2| \leq 1$. If $g_1 \in A_1$, as $m \geq 3$, $v_2(A_2) = 0 $ and $v_2(A_1 \setminus \{g\}) = 1$ and agent~$2$ envies agent~$1$ even after dropping the least-valued good (in agent~$2$'s view) from agent~$1$'s bundle, and hence $\mathcal{A}$ does not satisfy WEFX.
    Hence, the only WEFX allocation is $(G \setminus \{g_1\}, \{g_1\})$.
    \subsection{Proof of Proposition \ref{thm: WEFX, b,s}}
    Consider an instance with $n=2$ agents and $m=7$ identical goods and weights $w_1 = 2$ and $w_2 = 5$.  
    Let agents' valuations over goods be defined as follows: $v_1(S) = |S|$ and $v_2(S) = \min\{2,|S|\}$. As $v_1$ is additive and $v_2$ is budget-additive, both valuation functions are submodular. 
    
    In order for an allocation $A=(A_1, A_2)$ to satisfy WWEFX, for all $g \in A_2$, either $\frac{v_1(A_1)}{v_1(A_2 \setminus \{g\})} \geq 0.4$ or $\frac{v_1(A_1 \cup \{g\})}{v_1(A_2)} \geq 0.4$; and for all $g\in A_1$, either $\frac{v_2(A_2)}{v_2(A_1\setminus \{g \})} \geq 2.5$ or $\frac{v_2(A_2 \cup \{g \})}{v_2(A_1)} \geq 2.5$. 
    
    Consider any allocation $A=(A_1, A_2)$. We split our analysis into two cases.
    
    If $|A_1| \leq 1$, agent~$1$ envies agent~$2$ even after dropping the least-valued good (in agent~$1$'s view) from agent~$2$'s bundle, and the allocation is not WWEFX. To see this, $|A_1| = v_1(A_1) \leq 1$, $v_1(A_1 \cup \{g\}) \leq 2$, $v_1(A_2) \geq 7$  and $v_1(A_2 \setminus \{g\}) \geq 6$. Hence, $\frac{v_1(A_1)}{v_1(A_2 \setminus \{g\})} \leq 1/6 < 0.4$ and $\frac{v_1(A_1 \cup \{g\})}{v_1(A_2)} \leq 2/7 < 0.4$, violating the conditions we established earlier in order for an allocation to satisfy WWEFX.

    If $|A_1| \geq 2$, agent~$2$ envies agent~$1$ even after dropping the least-valued good (in agent~$2$'s view) from agent~$1$'s bundle, and the allocation is not WWEFX. To see this, $|v_2(A_2 \cup \{g\}) \geq v_2(A_2) \geq 2$ and $v_2(A_1) \geq v_2(A_1 \setminus \{g\}) \geq 1$. Hence, $\frac{v_2(A_2)}{v_2(A_1\setminus \{g \})} \leq 2 < 2.5$ and $\frac{v_2(A_2 \cup \{g \})}{v_2(A_1)} \leq 2 < 2.5$, again violating the conditions we established earlier in order for an allocation to satisfy WWEFX.
    
    \subsection{Proof of Proposition \ref{WEFX, restricted additive}}
    Consider an instance with $n=2$ agents and $m=4$ goods and weights $w_1 = 0.5 + \epsilon$ and $w_2 = 0.5 - \epsilon $. 
    Note that for $0<\epsilon < 1/14 $, $\frac{0.5 + \epsilon}{0.5 - \epsilon} > 1$ and $\frac{0.5 - \epsilon}{0.5 + \epsilon} > 3/4$. 
    Let agents' valuations over goods be defined as follows: $v_1(g_1) = v_1(g_2) = 2, v_1(g_3)=v_1(g_4) = 0$ and $v_2(g_1)=v_2(g_2) = 2, v_2(g_3) = 1, v_2(g_4) = 0$.
    
    Consider some allocation $\mathcal{A} = (A_1,A_2)$.
    If agent~$1$ receives both $g_1$ and $g_2$, then agent~$2$ will envy agent~$1$ even after dropping the least-valued good (in agent~$2$'s view) from agent~$1$'s bundle, and the allocation is not WWEFX.
    To see this, $v_2(A_1) \geq 4$, $v_2(A_1 \setminus \{g\}) \geq 2$, $v_2(A_2) \leq 1$, and $v_2(A_2 \cup \{g\}) \leq 3$.  Hence, $\frac{v_2(A_2)}{v_2(A_1\setminus \{g \})} \leq 1/2 $ and $\frac{v_2(A_2 \cup \{g \})}{v_2(A_1)} \leq 3/4$ and $\mathcal{A}$ satisfy WWEFX.
    
    Thus, we assume agent~$2$ gets at least one of $g_1$ and $g_2$.
    Without loss of generality, suppose agent~$2$ is allocated $g_1$. 
    If agent~$2$ receives only $g_1$, then agent~$2$ will envy agent~$1$ even after dropping the least-valued good (in agent~$2$'s view) from agent~$1$'s bundle, and the allocation is not WWEFX.
    To see this, $v_2(A_1) = v_2(A_1 \setminus \{g\}) = 3$ and $v_2(A_2\cup \{g\}) = v_2(A_2) = 2$ and $\frac{v_2(A_2)}{v_2(A_1\setminus \{g \})} = \frac{v_2(A_2 \cup \{g \})}{v_2(A_1)} = 2/3$. 
    If agent~$2$ is allocated $g_1$ and $|A_2| \geq 2$, then agent~$1$ will envy agent~$2$ even after dropping the least-valued good (in agent~$1$'s view) from agent~$2$'s bundle, and the allocation is not WWEFX.
    To see this, $v_1(A_2) \geq v_1(A_2 \setminus \{g\}) \geq 2$ and $v_1(A_1 \cup \{g\}) \geq v_1(A_1) \leq 2$.
    Hence, $\frac{v_1(A_1)}{v_1(A_2 \setminus \{g\})} \leq 1$ and $\frac{v_1(A_1 \cup \{g\})}{v_1(A_2)} \leq 1$ and $\mathcal{A}$ satisfy WWEFX.
\subsection{Proof of Theorem \ref{1/4-WEFX}}

Let agent~$1$ have weight $w_1$ and agent~$2$ have weight $w_2$. Normalize the weights and the agents' valuation functions (i.e., $w_1 + w_2 = 1$ and $v_1(g) = v_2(g) = 1$). Let $A \subseteq G$ be the bundle of goods such that the function $f(A) = \min (v_1(A) - w_1, 0) + \min (v_2(G \setminus A) - w_2, 0)$ is maximized. If there are ties, we pick the subset with the greatest cardinality.  If $f(A) = 0$, then $v_1(A) \geq w_1$ and $v_2(G \setminus A) \geq w_2$, and the allocation is WEFX. Otherwise, $v_1(A) < w_1$ and $v_2(G \setminus A) < w_2$ or $v_1(A) < w_1$, $v_2(G \setminus A) \geq w_2$ or $v_1(A) \geq w_1$, $v_2(G \setminus A) < w_2$. By renaming the agents, we can see that the latter two cases are identical, and thus we can solely consider the case $v_1(A) < w_1$, $v_2(G \setminus A) \geq w_2$.

We first prove several lemmas that will be used in both cases before delving into the two cases.
\begin{lemma}\label{aWEFX,lemma}
    (1) If there is a $g_1 \in G\setminus A$ such that $v_1(g_1) > v_2(g_1)$, then $v_1(A) + v_1(g_1) \geq w_1$. (2) If there is a $g_2 \in A$ such that $v_2(g_2) > v_1(g_2)$, then $ v_1(G \setminus A) + v_1(g_2) \geq w_2$.
\end{lemma}
\begin{proof}
     By renaming the agents, statement 2 is identical to statement 1, and thus it suffices to show statement 1. Suppose, $v_1(A) + v_1(g_1) < w_1$. We note that for all $g_2 \in G$, $\min\{v_2(G \setminus A) -w_2, 0\} - \min\{v_2(G \setminus (A \cup \{g_1\})) -w_2, 0\} <v_2(g_1)$. Then,
    \begin{align*}
        &f(A \cup \{g_1\}) - f(A) \\&= \min\{v_1(A \cup \{g_1\}) -w_1, 0\} - \min\{v_1(A \cup \{g_1\}) -w_1, 0\} \,\, +\\&  \min\{v_2(G \setminus (A \cup \{g_1\})\} -w_2, 0) - \min\{v_2(G \setminus A) -w_2, 0\}\\ 
        &\geq v_1(g_1) + \min\{v_2(G \setminus (A \cup \{g_1\}) -w_2, 0\} - \min\{v_2(G \setminus A) -w_2, 0\} \\&
 \text{ (as $v_1(A) + v_1(g_1) <w_1$ by assumption)} \\ 
 &\geq v_1(g_1) - v_2(g_1)  \\
        &> 0 \text{ (as } v_1(g_1) > v_2(g_2) \text{)}\\
    \end{align*}
Thus, contradicting our choice of $\mathcal{A}$.
\end{proof}
\begin{lemma}\label{aWEFX,lemma1}
    (1) If there is a $g_1 \in G\setminus A$ such that $v_1(g_1) > v_2(g_1)$, then $\min\{w_1 -v_1(A), v_1(g_1)\} \leq w_2$. (2) If there is a $g_2 \in A$ such that $v_2(g_2) > v_1(g_2)$, then $\min\{w_2 -v_2(G \setminus A), v_2(g_2)\} \leq w_1$.
\end{lemma}

\begin{proof}
    By renaming the agents, statement 2 is identical to statement 1, and thus it suffices to show statement 1. Suppose that $\min\{w_1 -v_1(A), v_1(g_1)\} > w_2$. Then,
    \begin{align*}
        &f(A \cup \{g_1\}) \\&= \min\{v_1(A \cup \{g_1\}) -w_1, 0\} + \min\{v_2(G \setminus (A \cup \{g_1\})\} -w_2, 0) \\ 
        &\geq \min\{v_1(A) + v_1(g_1) -w_1, 0\} - w_2 \\ 
        &= v_1(A) -w_1 + \min\{w_1 - v_1(A),v_1(g_1)\} - w_2 \\
        &> v_1(A) - w_1 \text{ (as } \min\{w_1 - v_1(A),v_1(g_1)\} > w_2 \text{ by assumption)}\\
        &\geq \min\{v_1(A) -w_1, 0\} + \min\{v_2(G \setminus A) -w_2, 0\}\\
        &= f(A).
    \end{align*}
    Thus, $f(A \cup \{g_1\}) > f(A)$, contradicting our choice of $A$.
\end{proof}

\begin{lemma}\label{aWEFX,lemma2}
     For an allocation $B$, (1) if $w_1 - v_1(B)  \leq w_2$ and $v_1(B) \geq w_1/2$, then $\frac{v_1(B)}{w_1} \geq \frac{1}{4}\frac{v_1(G \setminus B)}{w_2}$. (2) if $w_2 - v_2(G \setminus B)  \leq w_1$  and $v_2(G \setminus A) \geq w_2/2$, then $\frac{v_2(G \setminus B)}{w_2} \geq \frac{1}{4} \frac{v_2(B)}{w_1}$.
\end{lemma}
\begin{proof}
    Again by renaming the agents, statement 2 is identical to statement 1, and thus it suffices to show statement 1. We first note that $v_1(G \setminus B) = 1 - v_1(B) = w_2 + w_1 - v_1(b) \leq 2 w_2$. Thus, $\frac{v_1(B) / w_1}{v_1(G \setminus B) / w_2} \geq \frac{0.5}{2} = \frac{1}{4}$.
\end{proof}

With the lemmas above, we will delve into the two cases mentioned before.
\newline\newline
\textbf{Case 1:} $v_1(A) < w_1$ and $v_2(G\setminus A) < w_2$.

We first show that there is a $g_1 \in G\setminus A$ such that $v_1(g_1) > v_2(g_1)$. Otherwise, if for all $g \in A$, $v_1(g) \geq v_2(g)$, then $1 = v_2(G) = v_2(G \setminus A) + v_2(A) \leq v_2(G \setminus A) + v_1(A) < w_1 + w_2 = 1$ and we arrive at a contradiction. Through a similar logic, we can show that there is a $g_2 \in A$ such that $v_2(g_2) > v_1(g_2)$. Furthermore, by lemma \ref{aWEFX,lemma}, $v_1(A) + v_1(g_1) \geq w_1$ and  $v_2(G \setminus A) + v_1(g_2) \geq w_2$

If $v_1(A) \geq w_1/2$ and $v_2(A) \geq w_2/2$, then lemma \ref{aWEFX,lemma1} and lemma \ref{aWEFX,lemma2} together implies the allocation $(A, G\setminus A)$ is $\frac{1}{4}$-WEFX. Otherwise, without loss of generality, suppose that $v_1(A) < w_1/2$. As $v_1(A) + v_1(g_1) \geq w_1$, we know that $v_1(g_1) > w_1/2$. Thus, consider the allocation $(\{g_1\}, G \setminus \{g_1\})$. As agent~$1$ receives only 1 item, agent~$2$ would not envy agent~$1$ bundle after removing the item from agent~$1$ bundle. Furthermore, $w_1 - v_1(\{g\}) < w_1 - v_1(A) \text{ as $v_1(A) \leq v_1(\{g\}$) } \leq w_2$. Thus, by lemma \ref{aWEFX,lemma2}, $\frac{v_1(g)}{w_1} \geq \frac{1}{4}\frac{v_1(G \setminus \{g\})}{w_2}$. Thus, the allocation is $\frac{1}{4}$-WEFX.
\newline\newline
\textbf{Case 2:} $v_1(A) < w_1$ and $v_2(G \setminus A) \geq w_2$.

 First, if there is a $g_1 \in G\setminus A$ such that $v_1(g_1) > v_2(g_1)$, then through a similar reasoning as the the previous case, either the allocation $(A, G\setminus A)$ or $(\{g_1\}, G \setminus \{g_1\})$ is $\frac{1}{4}$-WEFX. Otherwise, for all $g \in G \setminus A$, $v_1(g) \leq v_2(g).$  If $|G \setminus A| = 1$, then the allocation $(A, G\setminus A)$ is trivially WEFX (and thus $\frac{1}{4}$-WEFX). Otherwise, $|G \setminus A| \geq 2$. We note that for all $g \in G \setminus A$, $v_2(G \setminus (A \cup \{g\})) < w_2$. Otherwise, $f(A \cup \{g\}) \geq f(A)$ and $|A \cup \{g\}| > |A|$, thus contradicting our choice of $A$. This leads us to 2 observations.
\begin{enumerate}
    \item For all $g \in G \setminus A$, $v_1(A \cup \{g\}) \geq w_1$. As for all $g \in G \setminus A$, $v_1(g) \leq v_2(g)$ and $v_2(G \setminus (A \cup \{g\})) < w_2)$, this implies that $v_1(G \setminus (A \cup \{g\})) < w_2$. Thus, $v_1(A \cup \{g\}) = 1 - v_1(G \setminus (A \cup \{g\})) = w_1 + w_2 -  v_1(G \setminus (A \cup \{g\})) \geq w_1$. 
    \item There exists a $g^* \in G\setminus A$ such that $v_2(g) < w_2$. This is because $|G \setminus A| \geq 2$ and for all $g \in G \setminus A$, $v_2(G \setminus (A \cup \{g\})) < w_2$.
\end{enumerate}

From observation 2, as $v_1(A \cup \{g^*\}) \geq w_1$, either $v_1(A) \geq w_1/2$ or $v_1(g^*) \geq w_1/2$. If, $v_1(A)\geq w_1/2$, we have that
$w_1 - v_1(A) < v_1(g^*) < w_2$ and by lemma \ref{aWEFX,lemma2}, the allocation $(A,G\setminus A)$ is $\frac{1}{4}$-WEFX. If $v_1(g^*)\geq w_1/2$, we have that $w_1 - v_1(g^*) \leq w_1/2 \leq v_1(g^*) < w_2$ and by lemma \ref{aWEFX,lemma2}, the allocation $(\{g^*\},G\setminus (\{g^*\})$ is $\frac{1}{4}$-WEFX.
\section{Omitted Proofs from Section 5}
\subsection{Proof of Theorem \ref{EFX+,n=2}}

    We first show that the first agent always has a valid cut.
    For any agent $i \in N$ with valuation function $v$, and bundles $A_1,A_2$, we say that $A_1 \prec_{leximax} A_2$ if and only if $v(A_1) > v(A_2)$, or $v(A_1) = v(A_2)$ and $|A_1| > |A_2|$. 

    For all allocations $A = (A_1,A_2)$, assume without loss of generality that $v_1(A_1) \leq v_1(A_2)$. 
    Then, agent~$1$ will select the partition such that $v(A_2)$ is the maximal (by $\prec_{leximax}$) across all possible partitions. 
    Suppose that agent~$1$ envies agent~$2$ even after adding some good from agent~$2$'s bundle if agent~$2$ receives $A_2$, i.e., there exists a $g \in A_2$ such that $v_1(A_1 \cup \{g\}) < v_1(A_2)$. 
    Consider the partition  $(A_1\cup \{g\},A_2 \setminus \{g\})$. 
    For $v_1$, $A_2 \prec_{leximax}  A_1\cup \{g\}$ and $A_2 \prec_{leximax}  A_2 \setminus \{g\}$, giving us a contradiction as $v(A_2)$ is not maximal.

    The argument for the lower bound on the number of EFX allocations when $n =2$ \cite{SUKSOMPONG2020606} extends to this setting as both employ a `cut-and-choose' algorithm. 
    The upper bound \cite{SUKSOMPONG2020606} also applies as the example given was under additive valuations, where both EFX and EFX+ are equivalent. Our result follows.
    \subsection{Proof of Proposition \ref{EFX+,3 agents 4 goods}}
    Consider an instance with $n=3$ agents and $m=4$ goods, and for all agents and bundles of goods $B_1,B_2$, $v(B_1) > v(B_2)$ if $|B_1| > |B_2|$. 
    Hence, if an agent receives no goods in an allocation, they would envy any agent that receive two goods even after adding the least-valued good (from the other agent's bundle) to their bundle. 
    Thus, we can only consider allocations where two agents receive one good each, and the last agent receives two goods. 
    As no agents will envy the agents with one good after adding that good to their bundle, the critical part of the valuation function is the ordinal rankings of bundles with two goods. 
    
    For agent~$1$,  $v(\{g_3, g_4\})< v(\{g_1, g_2\})< v(\{g_2, g_3\})< v(\{g_1, g_4\}) < v(\{g_2, g_4\})< v(\{g_1, g_3\})$. For agent~$2$,  $v(\{g_1, g_4\})< v(\{g_2, g_3\})< v(\{g_1, g_2\})< v(\{g_2, g_4\}) < v(\{g_3, g_4\})< v(\{g_1, g_3\})$. For agent 3,  $v(\{g_1, g_3\})< v(\{g_2, g_4\})< v(\{g_1, g_4\})< v(\{g_2, g_3\}) < v(\{g_3, g_4\})< v(\{g_1, g_2\})$.
    
    We then iterate over all 36 allocations (enumerated in Table \ref{EFX+_envy_analysis} below) where all agents receive at least one good; and we see that for each allocation, there is at least one agent that envies another agent even after adding the least-valued good (from the other agent's bundle) to their own bundle.
    
\begin{table}[H]
\centering
\caption{EFX+ envy analysis}
\label{EFX+_envy_analysis}
\begin{tabular}{ccc}
\hline
\textbf{Allocation} & \textbf{Envious Agent} & \textbf{Envied Agent} \\
\hline
\(\{g_1, g_2\}\), \(\{g_3\}\), \(\{g_4\}\) & 2 & 1 \\
\(\{g_1, g_2\}\), \(\{g_4\}\), \(\{g_3\}\) & 2 & 1 \\
\(\{g_1, g_3\}\), \(\{g_2\}\), \(\{g_4\}\) & 2 & 1 \\
\(\{g_1, g_3\}\), \(\{g_4\}\), \(\{g_2\}\) & 2 & 1 \\
\(\{g_1, g_4\}\), \(\{g_2\}\), \(\{g_3\}\) & 3 & 1 \\
\(\{g_1, g_4\}\), \(\{g_3\}\), \(\{g_2\}\) & 3 & 1 \\
\(\{g_2, g_3\}\), \(\{g_1\}\), \(\{g_4\}\) & 3 & 1 \\
\(\{g_2, g_3\}\), \(\{g_4\}\), \(\{g_1\}\) & 3 & 1 \\
\(\{g_2, g_4\}\), \(\{g_1\}\), \(\{g_3\}\) & 2 & 1 \\
\(\{g_2, g_4\}\), \(\{g_3\}\), \(\{g_1\}\) & 2 & 1 \\
\(\{g_3, g_4\}\), \(\{g_1\}\), \(\{g_2\}\) & 2 & 1 \\
\(\{g_3, g_4\}\), \(\{g_2\}\), \(\{g_1\}\) & 2 & 1 \\
\hline
\(\{g_3\}\), \(\{g_1, g_2\}\), \(\{g_4\}\) & 3 & 2 \\
\(\{g_4\}\), \(\{g_1, g_2\}\), \(\{g_3\}\) & 3 & 2 \\
\(\{g_2\}\), \(\{g_1, g_3\}\), \(\{g_4\}\) & 1 & 2 \\
\(\{g_4\}\), \(\{g_1, g_3\}\), \(\{g_2\}\) & 1 & 2 \\
\(\{g_2\}\), \(\{g_1, g_4\}\), \(\{g_3\}\) & 1 & 2 \\
\(\{g_3\}\), \(\{g_1, g_4\}\), \(\{g_2\}\) & 1 & 2 \\
\(\{g_1\}\), \(\{g_2, g_3\}\), \(\{g_4\}\) & 1 & 2 \\
\(\{g_4\}\), \(\{g_2, g_3\}\), \(\{g_1\}\) & 1 & 2 \\
\(\{g_1\}\), \(\{g_2, g_4\}\), \(\{g_3\}\) & 1 & 2 \\
\(\{g_3\}\), \(\{g_2, g_4\}\), \(\{g_1\}\) & 1 & 2 \\
\(\{g_1\}\), \(\{g_3, g_4\}\), \(\{g_2\}\) & 3 & 2 \\
\(\{g_2\}\), \(\{g_3, g_4\}\), \(\{g_1\}\) & 3 & 2 \\
\hline
\(\{g_3\}\), \(\{g_4\}\), \(\{g_1, g_2\}\) & 2 & 3 \\
\(\{g_4\}\), \(\{g_3\}\), \(\{g_1, g_2\}\) & 2 & 3 \\
\(\{g_2\}\), \(\{g_4\}\), \(\{g_1, g_3\}\) & 1 & 3 \\
\(\{g_4\}\), \(\{g_2\}\), \(\{g_1, g_3\}\) & 1 & 3 \\
\(\{g_2\}\), \(\{g_3\}\), \(\{g_1, g_4\}\) & 1 & 3 \\
\(\{g_3\}\), \(\{g_2\}\), \(\{g_1, g_4\}\) & 1 & 3 \\
\(\{g_1\}\), \(\{g_4\}\), \(\{g_2, g_3\}\) & 1 & 3 \\
\(\{g_4\}\), \(\{g_1\}\), \(\{g_2, g_3\}\) & 1 & 3 \\
\(\{g_1\}\), \(\{g_3\}\), \(\{g_2, g_4\}\) & 1 & 3 \\
\(\{g_3\}\), \(\{g_1\}\), \(\{g_2, g_4\}\) & 1 & 3 \\
\(\{g_1\}\), \(\{g_2\}\), \(\{g_3, g_4\}\) & 2 & 3 \\
\(\{g_2\}\), \(\{g_1\}\), \(\{g_3, g_4\}\) & 2 & 3 \\
\hline
\end{tabular}
\end{table}
\newpage
\section{Omitted Proofs from Section 6}
\subsection{Remark 1}

For the following instances, one can easily verify that these instances have exactly $n$ EFX allocations by iterating all $m ^n$ possible allocations.

For $n = 2, m = 4$,
\begin{table}[h!]
\centering
\begin{tabular}{c||c|c|c|c}
agent & $g_1$ & $g_2$ &$g_3$ & $g_4$ \\
\hline\hline
1 & 1 & 0 & 0  & 0 \\
2 &  1 & 0 & 0  & 0  \\
\end{tabular}
\end{table}

For $n = 3, m = 5$,
\begin{table}[h!]
\centering
\begin{tabular}{c||c|c|c|c|c}
agent & $g_1$ & $g_2$ &$g_3$ &$g_4$ &$g_5$   \\
\hline\hline
1 & 281 & 472 & 47  & 660 & 36 \\
2 & 569 & 936 & 173 & 343 & 135 \\
3 & 522 & 641 & 52  & 793 & 571 \\
\end{tabular}
\end{table}

For $n = 4, m = 6$,
\begin{table}[h!]
\centering
\begin{tabular}{c||c|c|c|c|c|c}
agent & $g_1$ & $g_2$ &$g_3$ &$g_4$ &$g_5$ &$g_6$   \\
\hline\hline
1 & 95  & 114 & 196 & 171 & 871 & 667 \\
2 & 254 & 973 & 200 & 240 & 907 & 536 \\
3 & 3   & 910 & 444 & 627 & 730 & 693 \\
4 & 382 & 651 & 425 & 182 & 548 & 811 \\
\end{tabular}
\end{table}

\end{document}